\documentclass{article}

\usepackage{arxiv}

\usepackage[utf8]{inputenc} 
\usepackage[T1]{fontenc}    
\usepackage{hyperref}       
\usepackage{url}            
\usepackage{booktabs}       
\usepackage{amsfonts}       
\usepackage{nicefrac}       
\usepackage{microtype}      
\usepackage{amsthm}
\usepackage{tabularx}
\usepackage{multirow}
\usepackage{multicol}

\usepackage{graphicx} 
\usepackage{caption}
\usepackage{subcaption}

\usepackage{natbib}
\usepackage{wrapfig}

\usepackage{algorithm}
\usepackage{algorithmic}

\usepackage{dsfont}

\usepackage{tikz}
\usetikzlibrary{bayesnet}
\usepackage{bm}
\usepackage{float}
\usepackage{comment}

\usepackage{mathtools}
\usepackage{amssymb}

\newtheorem{theorem}{Theorem}[section]

\newtheorem{lemma}[theorem]{Lemma}
 
\newcommand{\bo}[1]{\bm{#1}} 

\newcommand{\ba}{\begin{array}{rcll}} 
\newcommand{\ea}{\end{array}}

\newcommand{\ben}{\begin{enumerate}}
\newcommand{\bena}{\begin{enumerate}[(a)]}
\newcommand{\beni}{\begin{enumerate}[i)]}
\newcommand{\een}{\end{enumerate}}

\newcommand{\jen}[1]{{\color{blue}[\textbf{Jen:} #1]}}
\newcommand{\gui}[1]{{\color{green!50!black}[\textbf{Gui:} #1]}}
\newcommand{\vinayak}[1]{{\color{red}[\textbf{Vinayak:} #1]}}

\newcommand{\model}{\emph{Joint SBM}}
\newcommand{\isomodel}{\emph{Isolated SBM}}
\newcommand{\jspec}{\emph{JointSpec}}
\newcommand{\isospec}{\emph{IsoSpec}}

\newcommand{\stacked}[1]{\overline{\overline{#1}}}
\newcommand{\diag}{\text{diag}}

\usepackage{makecell}

\DeclareMathOperator*{\argmin}{arg\,min}

\title{Community detection over a heterogeneous population of non-aligned networks}

\author{
  Guilherme Gomes$^1$, Vinayak Rao$^{1,2}$, Jennifer Neville$^{1,2}$ \\
$^1$Department of Statistics, Purdue University\\ 
$^2$Department of Computer Science, Purdue University \\
  \texttt{\{gomesg, varao, neville\}@purdue.edu} \\
}

\begin{document}
\maketitle
\vspace{-7mm}
\begin{abstract}
Clustering and community detection with multiple graphs have typically focused on {\em aligned} graphs, where there is a mapping between nodes across the graphs (e.g., multi-view, multi-layer, temporal graphs). However, there are numerous application areas with multiple graphs that are only partially aligned, or even unaligned. These graphs are often drawn from the same population, with communities of potentially different sizes that exhibit similar structure. In this paper, we develop a joint stochastic blockmodel (\model) to estimate shared communities across sets of heterogeneous non-aligned graphs. We derive an efficient spectral clustering approach to learn the parameters of the joint SBM\footnote{Code available at \href{https://github.com/kurtmaia/JointSBM}{github.com/kurtmaia/JointSBM}.}. We evaluate the model on both synthetic and real-world datasets and show that the joint model is able to exploit cross-graph information to better estimate the communities compared to learning separate SBMs on each individual graph.
\end{abstract}

\vspace{-3mm}
\section{Introduction}

Interest in the statistical analysis of network data has  grown rapidly over the last few years, with applications such as link prediction \citep{liben2007link,wang2015link}, collaborative filtering \citep{sarwar2001item,linden2003amazon,koren2015advances},
community detection \citep{fortunato2010community,aicher2014learning,newman2016community} etc.
In the area of community detection, most work has focused on single large graphs. Efforts to move beyond this to consider scenarios with multiple graphs, such as multi-view clustering~\citep{cozzo2013clustering}, multi-layer community detection \citep{mucha2010community,zhou2017graph}, and temporal clustering \citep{von2016mobilitygraphs}. However, these efforts assume that the graphs are {\em aligned}, with a known correspondence or mapping between nodes in each graph.  Here, the multiple graphs provide different sources of information about the same set of nodes.  Applications of aligned graphs include graphs evolving over time \citep{peixoto2015inferring, sarkar2012nonparametric,charlin2015dynamic,durante2014nonparametric}; as well as independent observations \citep{ginestet2014,asta2014geometric, durante2015,durante2016}. 

Such methods are not applicable to multiple non-aligned graphs, which occur in   
fields like biology (e.g., brain networks, fungi networks, protein networks), 
social media (e.g., social networks, word co-occurrences), and many others. In these settings, graph instances are often drawn from the same population, with limited, or often no, correspondence between the nodes across graphs (i.e., the nodes have similar behavior but represent different entities). Clustering and community detection methods for this scenario are relatively under-explored. While one could cluster each graph separately, pooling information across graphs can improve estimation, particularly for graphs with sparse connectivity or imbalanced community sizes. 

Here, we focus on the problem of community detection across a set of non-aligned graphs of varying size. 
Formally, we are given a set of $N$ graphs, $\bm{\Omega} = \{\bm{A}_1,\dotsc,\bm{A}_N\}$, where the $n$th graph is represented by its adjacency matrix defined as $\bm{A}_n\in \{0,1\}^{|V_n|\times|V_n|}$ where ${V}_n$ is the set of nodes. Write ${E}_n$ for the edges. We do not require the set $V_n$ to be shared across different graphs, however we assume they belong to a common set of $K$ communities. %
Our goal is to identify the $K$ communities underlying $\bo{\Omega}$. 
For instance, consider villages represented by social graphs where nodes represent individuals in a village and edges 
represent relationships between them. Although the people are different in each village and the sizes of villages vary, personal characteristics may impact the propensity of having some type of relationship, e.g. a younger individual is more willing to connect with other younger individuals,
or an influential person such as a priest %
might be expected to have more ties. 
Indeed, there are a vast number of factors, both observed and latent, shared among people across villages, that 
might influence relationship and  community formation. %

Our approach builds on the popular stochastic block model (SBM). SBMs have been studied extensively for single graph 
domains~\citep{rohe2011spectral,lei2015consistency,sarkar2015role}, and have proven to be highly effective on real 
world problems~\citep{airoldi2008mixed,karrer2011stochastic}. %
A simple approach to apply SBMs in our setting is to separately estimate the SBM parameters per graph, and then attempt to find a correspondence between the estimated probability structures to determine a single {\em global} community structure. We call this process \isomodel\ and show it is only reasonably accurate when each community is well represented in each graph.  Moreover, the process of finding a correspondence among the estimated probabilities has $\mathcal{O}(NK^2)$ complexity, which is only feasible when $N$ is small and $K$ is very small.

We propose a joint SBM model, where the key issue is to estimate a joint connectivity structure among the communities in each graph while allowing the number and sizes of clusters to vary across graphs. To estimate the parameters of the model, we show how the individual graph adjacency matrix eigendecompositions relate to the decomposition of the whole dataset $\bm{\Omega}$, and then derive an efficient spectral clustering approach to optimize the joint model. Our learning algorithm has complexity $\mathcal{O}(|V|K)$ per iteration where $|V|$ is total number of nodes across graphs. Notably, our approach is more efficient than \isomodel\ because it does not need the $\mathcal{O}(NK^2)$ step to determine a correspondence between the separate models across graphs. We evaluate our \model\ on synthetic and real data, comparing the results to the Isolated SBM and an alternative approach based on clustering graph embeddings.  The results show that our joint model is able to more accurately recover the community structure, particularly in scenarios where graphs are highly heterogeneous.
%

%

%

%
\section{Joint SBM for multiple graphs}

\subsection{Single graph SBM}

The stochastic blockmodel (SBM)~\citep{holland1983stochastic, wasserman1987stochastic} can be viewed as a mixture of Erd\"os-R\'enyi graphs~\citep{erdds1959random}, and for a single graph $\bm{A}$
with $K$ communities, is defined by a $(|V_n| \times K)$ membership matrix $\bm{X}$, and a connectivity matrix $\bm{\Theta}  \in [0,1]^{K \times K}$.  
Here $\bm{X}[i,k]=1$ if the $i$-th node belongs to community $k$, and equals $0$ otherwise.  
$\bm{\Theta}[k,l] = \theta_{kl}$ is the probability of an edge between nodes from communities $k$ and $l$. 
Then, a graph represented by adjacency matrix $\bm{A}$ is generated as
\vspace{-7mm}

\begin{align}
a_{ij} \sim
  \text{Bern}(\bm{X}_i\bm{\Theta}\bm{X}_j^{T})  &\quad \text{if $i < j$}.
\label{eq:gen_model_single}
\end{align}

\vspace{-4mm}
Note that $a_{ij} = a_{ji}\ \forall i,j$, and since we do not consider self-edges, $a_{ii} = 0\ \forall i$.
SBMs have traditionally been applied in settings with a single graph. 
Here, theoretical properties like consistency and goodness-of-fit are well understood, and efficient polynomial-time algorithms with theoretical guarantees have been proposed for learning and inference. 
However, there is little work for the situation with multiple observed graphs with shared statistical properties. 

\subsection{Multi-graph joint SBM}
\vspace{-1mm}
To address this, we consider an extension of the SBM in Equation~\eqref{eq:gen_model_single}.
Our model does not require vertices to be aligned across graphs, nor does it require different graphs to have the same number of vertices.
Vertices from all graphs belong to one of shared set of groups, with membership of graph $n$ represented by a membership matrix 
$\bm{X}_{n}$.  
Edge-probabilities between nodes are determined by a {\em global} connectivity matrix $\bm{\Theta}$ shared by all graphs. 
For notational convenience, we will refer to the set of stacked $\bm{X}_{n}$ matrices as the full membership matrix $\bm{X}$, with $\bm{X}_{ni}$ one-hot membership vector of the $i$-th node of graph $n$.  
The overall generative process assuming $K$ blocks/clusters/communities is 

\vspace{-4mm}
\begin{equation}
\begin{aligned}[l]
a_{nij} \sim
  \text{Bern}(\bm{X}_{ni}\bm{\Theta}\bm{X}_{nj}^{T})  &\quad \text{if $i < j$} \\
\end{aligned}
\label{eq:gen_model}
\end{equation}
where again, $a_{nij}$ is $ij$-th cell of the adjacency matrix $\bo{A}_n$.
Note that $\bm{X}_{n}$ is a $(|V_n| \times K)$ binary matrix, and $\bm{\Theta}$ is a ${K \times K}$ matrix of probabilities. 
This model can easily be extended to edges with weights (replacing the Bernoulli distribution with some other distribution) or to include covariates (e.g.\ through another layer of coefficients relating covariates with membership or edge probabilities).

\section{Inference via spectral clustering}
Having specified our model, we provide a spectral clustering algorithm to identify community structure and edge probabilities from graph data $\bo{A}_1, \bo{A}_2, \ldots, \bo{A}_N$. 

\vspace{-1mm}
\subsection{Spectral clustering for a single graph}
First, we recall the spectral clustering method to learn SBMs for a single graph $\bo{A}_n$~\citep{von2007tutorial,lei2015consistency}.
Let $\bo{P}_n$ refer to the edge probability matrix of graph $n$ under an SBM, where $\bo{P}_n = \bm{X}_{n}\mathbf{\Theta}\bm{X}_{n}^{T}$ (Equation~\eqref{eq:gen_model_single}). 
Since we do not consider self-loops, $\mathbb{E}[\bo{A}_n] = \bo{P}_n - \text{diag}(\bo{P}_n)$. %
Further, write the eigendecomposition of $\bo{P}_n$ as $\bo{P}_n = \bo{U}_n\bo{D}_n\bo{U}_n^{T}$. Here, $\bo{U}_n$ is a $(|V_n|\times K)$ matrix of eigenvectors related to the $K$ largest absolute eigenvalues and $\bo{D}_n$ is a $K \times K$ diagonal matrix with the $K$ non-zero eigenvalues of $\bm{P}_n$. 
Let $|G_{nk}|$ refer to the number of nodes that are members of cluster $k$, and $\Delta_n  =\left(\bm{X}_{n}^{T}\bm{X}_{n}\right)^{1/2}$ define a $K \times K$ diagonal matrix with entries $\sqrt{|G_{nk}|}$.
Define $\bo{Z}_n\tilde{\bo{D}}_n\bo{Z}_n^{T}$ as the eigendecomposition of $\Delta_n\mathbf{\Theta}\Delta_n$. Then
\begin{align}
  \bo{U}_n\bo{D}_n\bo{U}_n^{T} = \bo{P}_n  &= \bm{X}_{n}\mathbf{\Theta}\bm{X}_{n}^{T} 
 = \bm{X}_{n}\Delta_n^{-1}\Delta_n\mathbf{\Theta}\Delta_n\Delta_n^{-1}\bm{X}_{n}^{T} = \bm{X}_{n}\Delta_n^{-1}\bo{Z}_n\tilde{\bo{D}}_n\bo{Z}_n^{T}\Delta_n^{-1}\bm{X}_{n}^{T}.
\label{eq:lemma_proof}
 \end{align}
 Since $\bo{D}_n$ and $\tilde{\bo{D}}_n$ are both diagonal, and  $\bm{X}_{n}\Delta_n^{-1}$ and $\bo{Z}_n$ are both orthonormal, 

\vspace{-5mm}
\begin{equation}
  \bo{D}_n = \tilde{\bo{D}}_n, \quad \bo{U}_n = \bm{X}_{n}\Delta_n^{-1}\bo{Z}_n.
\label{eq:eigen_def_single}
\end{equation}
In practice, we use the observed adjacency matrix $\bo{A}_n$ as a proxy 
for $\bo{P}_n$, and replace $\bo{U}_n$ in Eqs.~\eqref{eq:lemma_proof} and \eqref{eq:eigen_def_single}
with $\widehat{\bo{U}}_n$ calculated from the eigendecomposition 
$\bo{A}_n = 
\widehat{\bo{U}}_n\widehat{\bo{D}}_n\widehat{\bo{U}}_n^{T}$. 
Finally we note that each row of $\bm{X}_{n}$ has only one
non-zero element, indicating which group that node belongs to. Thus, as in \cite{lei2015consistency}, we can 
use k-means clustering to recover $\bm{X}_{n}$ and $\bo{W}_n= \Delta_n^{-1}\bo{Z}_n$ from 
$\widehat{\bo{U}_n}$:

\vspace{-5mm}
\begin{align}
  \left(\widehat{\bm{X}}_{n}, \widehat{\bm{W}}_n\right) = \argmin_{\substack{\bm{X}_{n} \in \mathcal{M}_{|V_n|,K}\\\bm{W}_n \in \mathbb{R}^{K \times K}}} \|\bm{X}_{n} \bm{W}_n - \widehat{\bm{U}}_n\|_F^2
\label{eq:opt_kmeans}
\end{align}

\vspace{-4mm}
Given a solution $\left(\widehat{\bm{X}}_{n}, \widehat{\bm{W}}_n\right)$ from Equation~\eqref{eq:opt_kmeans}, we can estimate $\bo{\Theta}$ from the cluster memberships $\widehat{\bo{X}}_n$ as:
\begin{align}
\widehat{\bm{\Theta}}_n =& \widehat{\bm{S}}_n + \widehat{\Delta}_n^{-2}\left[\mathbb{I}_k-\widehat{\Delta}_n^{-2}\right]^{-1}\diag\left(\widehat{\bm{S}}_n\right)
\label{eq:theta_estimator_n}
\end{align}

\vspace{-2.8mm}
where $\widehat{\bm{S}}_n = \widehat{\Delta}_n^{-2}\widehat{\bm{X}}_n^T\bm{A}_n\widehat{\bm{X}}_n\widehat{\Delta}_n^{-2}$ and $\widehat{\Delta}_n^{2} = \widehat{\bm{X}}_n^T\widehat{\bm{X}}_n$.
If we assume the true membership matrix $\bm{X}$ is known, then the estimator $\widehat{\bm{\Theta}}$ in Eq.~\eqref{eq:theta_estimator_n} is unbiased for the true $\bm{\Theta}$. See Appendix \ref{subsec:theta_n_details} for details of derivation and unbiasedness. %
\vspace{-2mm}
\subsection{Naive spectral clustering for multiple graphs }
\vspace{-2mm}
Given multiple unaligned graphs, the procedure above can be applied to each graph, returning a set of $\bm{X}_{n}$'s and $\bm{\Theta}_n$'s, one for each graph. The complexity of this is $\mathcal{O}(N\phi + |V|K)$, where $\phi$ refers to the complexity of eigen decomposition on a single graph (typically $\mathcal{O}(|E|)$ for sparse graphs \cite{pan1999complexity,golub2012matrix}). 
However, this does not recognize that a single $\bm{\Theta}$ is shared across all graphs. 
Estimating a global $\bm{\Theta}$ from the individual $\bm{\Theta}_n$s requires an {\em alignment} step, to determine a mapping among the $\bm{\Theta}_n$s. The complexity of determining the alignment is $\mathcal{O}(NK^2)$ and is a two-stage procedure that results in loss of statistical efficiency, especially in 
settings with heterogenerous, imbalanced graphs. 
We refer to this approach as \isomodel.
We propose a novel algorithm to get around these issues by understanding how each graph relates to the global structure. 

\vspace{-2mm}
\subsection{Joint spectral clustering for multiple graphs}
\vspace{-2mm}

Let $|V| = \sum_n |V_n|$, where $|V_n|$ refers to the number of nodes in $A_n$. 
Consider $|V| \times |V|$ block diagonal matrix $\bo{A}$ representing the whole dataset of adjacency matrices, and define an associated probability matrix $\bo{P}$:

%
\vspace{-2mm}
\begin{equation}
\hspace{-2mm}
\bo{A} = 
  \begin{bmatrix}
\bo{A}_1 & \dots & 0 & \dots & 0\\
\vdots & \ddots & \vdots & \vdots & \vdots  \\
0 & \dots & \bo{A}_n & \dots & 0 \\
\vdots & \vdots & \vdots & \ddots&\vdots\\
0 & \dots & 0 & \dots & \bo{A}_N\\
\end{bmatrix},
\bo{P}= %
\begin{bmatrix}
\bo{P}_1 & \dots & \bo{P}_{1n} & \dots & \bo{P}_{1N}\\
\vdots &  \ddots & \vdots & \vdots & \vdots \\
\bo{P}_{n1} & \dots & \bo{P}_n & \dots & \bo{P}_{nN} \\
\vdots &  \vdots & \vdots & \ddots & \vdots \\
\bo{P}_{N1} &  \dots & \bo{P}_{Nn} & \dots & \bo{P}_N\\
\end{bmatrix} .
\nonumber
\end{equation} 

%

\vspace{-2mm}
Write the membership-probability decomposition of $\bo{P}$ as $\bo{P} = \bm{X}\mathbf{\Theta}\bm{X}^{T}$, here, $\bm{X}$ is the stacked $|V| \times K$ matrix of the $\bm{X}_{n}$'s for all graphs, and $\bo{\Theta}$ a $K \times K$ matrix of edge-probabilities among groups.
Note that for $i \neq j$, $\bo{P}_{ij}$ includes edge-probabilities between nodes {\em in different graphs}, something we cannot observe. 
As before, define $\Delta = \left(\bm{X}^{T}\bm{X}\right)^{1/2}$,
and the eigendecomposition of $\bo{P}$ gives 
\begin{align}
\bo{U}\bo{D}\bo{U}^{T} = \bo{P} 
 &= \bm{X}\mathbf{\Theta}\bm{X}^{T} 
 = \bm{X}\Delta^{-1}\Delta\mathbf{\Theta}\Delta\Delta^{-1}\bm{X}^{T} \nonumber = \bm{X}\Delta^{-1}\bo{Z}\tilde{\bo{D}}\bo{Z}^{T}\Delta^{-1}\bm{X}^{T},  
\label{eq:lemma_proof2}
\end{align}
with $\bo{Z}\tilde{\bo{D}}\bo{Z}^{T}$ corresponding to the eigendecomposition of $\Delta\mathbf{\Theta}\Delta$. Thus, similar to the single graph case 
\begin{equation}
 \bo{D} = \tilde{\bo{D}},  \quad \bo{U} = \bm{X} \Delta^{-1}\bo{Z}, 
\label{eq:eigen_def_multiple}	
\end{equation}
Note that $\bo{D}$ is still a $K \times K$ diagonal matrix. Let $\bm{U}_{n*} = \bm{X}_{n*} \Delta^{-1}\bo{Z}$ refer to the subset of $\bm{U}$ corresponding to the nodes in graph $n$, and define $\bm{X}_{n*}$ similarly. 
Note that $\bm{X}_{n*} = \bm{X}_{n}$, though $\bm{U}_{n*}$ differs from $\bm{U}_{n}$ of Eq.~\eqref{eq:eigen_def_single}.
By selecting the decomposition related to graph $n$, we have,  
\begin{align}
\bo{P}_n = \left(\bo{P}\right)_{n,n} = \left(\bo{U}\bo{D}\bo{U}^{T}\right)_{n,n}  
= \bo{U}_{n*}\bo{D}\bo{U}_{n*}^{T}.
\label{eq:stacked2partial}
\end{align}
From Eq.\eqref{eq:lemma_proof}, %
$\bo{U}_{n*}\bo{D}\bo{U}_{n*}^{T}=\bo{U}_n\bo{D}_n\bo{U}_n^{T}$.
If we let $\bo{Q}_n:=\bo{U}_n\bo{D}_n$,
\begin{align}
\bo{U}_{n*}\bo{D}\bo{U}_{n*}^{T} &=\bo{Q}_n\bo{U}_n^{T} \nonumber \\
\bo{U}_{n*}\bo{D}\bo{U}_{n*}^{T}\bo{U}_n &=\bo{Q}_n\bo{U}_n^{T}\bo{U}_n \nonumber \\
\bo{U}_{n*}\bo{D}(\bm{X}_{n*}\Delta^{-1}\bo{Z})^{T}\bm{X}_{n} \Delta_n^{-1}\bo{Z}_n &=\bo{Q}_n \nonumber \\
\bo{U}_{n*}\bo{D}\bo{Z}^{T}\Delta^{-1}\bm{X}_{n*}^{T}\bm{X}_{n} \Delta_n^{-1}\bo{Z}_n &=\bo{Q}_n \nonumber \\
\bo{U}_{n*}\bo{D}\bo{Z}^{T}\Delta^{-1}\Delta_n^{2} \Delta_n^{-1}\bo{Z}_n &=\bo{Q}_n \label{eq:q_n} 
\end{align}
From Eq.~\eqref{eq:eigen_def_multiple}, we then have 
\begin{equation}\bo{U}_{n*}\bo{D} =\bo{Q}_n \bo{Z}_n^T \Delta_n^{-1} \Delta \bo{Z}=\bm{X}_{n*}\bm{W}
\label{eq:global_local_expression}
\end{equation}
where $\bm{W}:=\Delta^{-1}\bo{Z} \bo{D}$.
Contrast this with the single graph, which from Eq.\eqref{eq:eigen_def_single}, gives $\bo{U}_n\bo{D}_n=\bo{Q}_n=\bm{X}_{n}\Delta_n^{-1}\bo{Z}_n\bo{D}_n$. 
If we can estimate the middle (or right) term in Eq.~\eqref{eq:global_local_expression} from data, 
we can solve for joint community assignments: %
\begin{equation}
\begin{aligned}
\left(\widehat{\bm{X}}, \widehat{\bm{W}}\right) 
=\hspace{-2mm}\argmin_{\substack{\bm{X} \in \mathcal{M}_{|V|,K} \\ \bm{W} \in \mathbb{R}^{K \times K}}} 
\sum_n^N \|\bm{X}_{n} \bm{W} - {\bo{Q}}_n \bm{Z}_n^{T}\Delta_n^{-1}\Delta\bo{Z} \|_F^2
\label{eq:opt_kmeans_graphLevel}
\end{aligned}
\end{equation}
While we can estimate $\bo{Q}_n$ from the data (Equation~\eqref{eq:q_n}), 
we cannot estimate $\bo{Z}_n$ or $\bo{Z}$ trivially. Instead, we will optimize an upper bound of a transformation of Equation~\eqref{eq:opt_kmeans_graphLevel}.
\begin{lemma}
The optimization in Eq.~\eqref{eq:opt_kmeans_graphLevel} is equivalent to 
\begin{equation}
  \hspace{-2mm}\argmin_{\substack{\bm{X} \in \mathcal{M}_{|V|,K} \\ \bm{W} \in \mathbb{R}^{K \times K}}}  \sum_{n=1}^N \left\|a_n(\bm{X}_n,\bm{W})  + b_n(\bm{X}_n) \right\|^2_F, \text{where}
\label{eq:opt_kmeans_graphLevel_v2}
\end{equation}
\vspace{-5mm}
$$a_n(\bm{X}_n,\bm{W}) := \left( \bm{X}_{n}\bo{W} - \bo{Q}_n^* \right)\bm{Z}^{T}\Delta^{-1}\Delta_n\bm{Z}_n, \ \ \ \ \  \ \ \  
b_n(\bm{X}_n) := \bo{Q}_n^*\left(\bm{Z}^{T}\Delta^{-1}\Delta_n\bo{Z}_n - \sqrt{\frac{|V_n|}{|V|}}\mathbb{I}_K\right)$$
\label{lem:kmean_opt}
\end{lemma}
\vspace{-7mm}
\begin{proof}
See Appendix \ref{subsec:opt_kmeans_graphLevel_details_v2}.
\end{proof}

Next we derive a bound for Eq.~\eqref{eq:opt_kmeans_graphLevel_v2} 
using the triangle inequality and sub-multiplicative norm property. Let $|\bm{M}|$ be the element-wise absolute values of matrix $\bm{M}$. Then 

\begin{lemma} The following inequality holds
\begin{align}
  \frac{1}{2}\left\|a_n(\bm{X}_n,\bm{W})  + b_n(\bm{X}_n) \right\|^2_F  &
  \le  \left\| \bm{X}_{n}\bo{W} \!\!-\! \bo{Q}_n^* \right\|_F^2 \gamma_n + \left\| \left|\bo{Q}_n^*\right| \left( \Delta^{-1}\Delta_n + \sqrt{\frac{|V_n|}{|V|}} \right) \right\|_F^2
\label{eq:opt_kmeans_graphLevel_v3} \\ 
  := & \widetilde{a}_n(\bm{X}_n,\bm{W}) + \widetilde{b}_n(\bm{X}_n) 
:= \eta_n(\bm{X}_n,\bm{W})\nonumber
\end{align}
\vspace{-4mm}
\begin{align}
 \text{where  }  \gamma_n &= \left\|\bm{Z}\Delta^{-1}\Delta_n\bo{Z}_n^{T}\right\|_F^2 = \text{tr}\left(\bm{Z}_n\Delta_n\Delta^{-2}\Delta_n\bm{Z}_n^{T}\right) \nonumber \\
        &= \text{tr}\left(\Delta_n^2\Delta^{-2}\right)=\sum_{m=1}^K \frac{|G_{n m}|}{|G_{\cdot m}|}. \label{eq:def_gamman}
\end{align}
\vspace{-4mm}
\label{lem:kmean_opt2}
\end{lemma}
\begin{proof}
  See Appendix~\ref{subsec:opt_kmeans_graphLevel_details}
\end{proof}
\vspace{-3mm}
%
\paragraph{Summary:} We can now optimize the bound on Eq.~\eqref{eq:opt_kmeans_graphLevel_v2} 

\begin{equation}
  \left(\widehat{\bm{X}}, \widehat{\bm{W}} \right) = 
\argmin_{\substack{\bm{X} \in \mathcal{M}_{|V|,K} \\ 
\bm{W} \in \mathbb{R}^{K \times K}}}
\sum_n^N \eta_n(\bm{X}_n,\bm{W})
\label{eq:kmeans_final}
\end{equation}
The terms $\widetilde{a}_n(\cdot)$ and $\widetilde{b}_n(\cdot)$  are weighted  sums of squares of $\bm{Q}_n$. However, we center each $\bm{Q}_n$ at the global weighted mean $\bm{W}$  in $\widetilde{a}_n(\cdot)$. The term $\gamma_n$ controls the importance of the global parameter $\bm{W}$ in each graph. Thus, $\gamma_n$ downweights the effect of $\bm{W}$ in small graphs and in graphs with highly underrepresented communities.  Intuitively, the term $\widetilde{a}_n(\cdot)$ is assigning nodes to clusters assuming $\bm{Z}_n^{T}\Delta_n^{-1}\Delta\bo{Z} = \sqrt{|V|/|V_n|}\mathbb{I}_K$, and  $\widetilde{b}_n(\cdot)$ accounts for the distance between a given graph and the global distribution of nodes over clusters.

\paragraph{Optimization:} We optimize Equation~\eqref{eq:kmeans_final} using a heuristic inspired by Lloyd's algorithm for k-means. 
This involves iterating two steps: (1) compute the means $\bm{W}$ given observations in each cluster $\bm{X}$; (2) assign observations to clusters $\bm{X}$ given means $\bm{W}$:
\begin{enumerate}
   \item \textbf{Compute the means:} We update $\bm{W}$ by minimizing $\eta_n(\bm{X_n},\bm{W})$ for a given $\bm{X}$, i.e. $\sum_n^N \nabla_{_{\bm{W}}} \eta_n(\bm{X}_n,\bm{W}) = 0$. Note that the $\widetilde{b}_n$ does not involve $\bm{W}$ so it is dropped,
\begin{align}
&\sum_n^N 2\bm{X}_n^T(\bm{X}_n\widehat{\bm{W}} - \bm{Q}^*_n)\gamma_n = 0 \nonumber \\
&\widehat{\bm{W}} = \left[\sum_n^N \bm{X}_n^T\bm{X}_n\gamma_n\right]^{-1} \sum_n^N \bm{X}_n^T\bm{Q}^*_n\gamma_n  
\label{eq:w_update_rule}
\end{align}
   \item \textbf{Assign nodes to communities:} We assign each node to the cluster that minimizes Eq.~\eqref{eq:kmeans_final}.
     Accordingly, define $\omega_{ni}(k)$ as the distance of node $i$ to cluster $k$: %
\begin{align}
  \omega_{ni}(k) = \left\|\bo{W}_k - \bo{Q}_{ni}^* \right\|^2\text{tr}\left(\widetilde{\Delta}_{ni}^2(k) \right) + \left\|\left|\bo{Q}_{ni}^*\right|\left(\widetilde{\Delta}_{ni}(k) + \sqrt{\frac{|V_n|}{|V|}}\mathbb{I}_K \right) \right\|^2
\label{eq:def_omega2}
\end{align}
\noindent $\bm{W}_k$ is the $k$-th row of $\bm{W}$ and $\widetilde{\Delta}_{ni}(k))$ is the value of $\Delta^{-1}\Delta_n$ if node $i$ is placed in cluster $k$. Precisely, say node $i$ is currently in cluster $l$, then we have
\begin{align}
\widetilde{\Delta}_{ni}(k) = \left[(\Delta^{2}-\diag(H_l)+\diag(H_k))\right]^{-1/2} \times \left[(\Delta_n^{2}-\diag(H_l)+\diag(H_k))\right]^{1/2}
\label{eq:def_gammaTilde}
\end{align}
\noindent where $H_l$ is a size $K$ one-hot vector at position $l$. Then
\begin{align}
\bm{X}_{ni} =& \argmin_k \omega_{ni}(k)
\end{align}
 \end{enumerate} 
 \vspace{-3mm}
\paragraph{Algorithm:} Algorithm~\ref{algo:jointSpec2} outlines the overall procedure for learning the \model. The complexity is $\mathcal{O}(N\phi + |V|K)$. Recall that $\phi$ refers to the complexity of eigen decomposition on a single graph. Note that our derived objective does not require decomposition of the full graph $\bo{A}$, instead decomposing each individual graph and then using the results to jointly estimate $\bo{X}$ and $\bo{W}$.
Notice that the extra $\mathcal{O}(NK^2)$ for alignment in the \isomodel\ is not needed here.  
Given the cluster assignments $\bm{X}$, we can easily estimate the cluster edge probabilities $\bm{\Theta}$ (see Eq.\eqref{eq:theta_estimator_n}). 
%

%
%
%

%
%
%
%
%
%
%
%

%

%
%
%
%
%
%
%
%
%
%

%
%
\vspace{-2mm}
\section{Comparing Joint SBM to Isolated SBM}
\label{sec:analysis}

Both \model\ and \isomodel\ use the eigen decompositions of the individual $N$ graphs, and are closely related.  
The biggest difference is the $\bm{Z}_n^{T}\Delta_n^{-1}\Delta\bo{Z}$ term in the definition of $\widetilde{a}_n(\bo{X}_n,\bo{W})$ (Eq.~\eqref{eq:opt_kmeans_graphLevel_v3}), which, intuitively, normalizes each individual graph eigenvector based on the distribution of nodes over the clusters in the graph.  
If the graphs are balanced, i.e., they have roughly the same proportion of nodes over clusters, then $\bm{Z}_n^{T}\Delta_n^{-1}\Delta\bo{Z} \approx \sqrt{|V_n|^{-1}|V|}$ which does not depend on the cluster sizes. Lemma \ref{lemma:same_prop_pop} formalizes this (See Appendix \ref{subsec:proof_lemma} for proof):
\begin{lemma}
Let the pair $(\bm{X},\bm{\Theta})$ represent a joint-SBM with $K$ communities for $N$ graphs, where $\bm{X}$ is the stacked membership matrix over the $N$ graphs and $\bm{\Theta}$ is full rank. The size of each graph $|V_n|$ may vary, but assume the graphs are balanced in expectation in terms of communities, i.e., assume the same distribution of cluster membership for all graphs: $\bm{X}_{ni} \stackrel{iid}{\sim} \text{Mult}(\bm{\zeta})$ for all $n \in [1,...,N]$ and $i \in [1,...,|V_n|]$. Then, 
\begin{equation}\mathbb{E}\left[\bm{Z}_n^{T}\Delta_n^{-1}\Delta\bo{Z}\right] = \sqrt{|V_n|^{-1}|V|}
\end{equation}
\label{lemma:same_prop_pop}
\end{lemma}
\vspace{-8mm}
For cases where we expect Lemma~\ref{lemma:same_prop_pop} to be true, 
we have:
\begin{align*}
\mathbb{E}\left[\eta_n(\bm{X}_n,\bm{W})\right] = \left\| \bm{X}_{n}\bo{W} - \bo{Q}_n^* \right\|_F^2\frac{|V_n|K}{|V|}  
+ \left\| 2\left|\bo{Q}_n^*\right|\sqrt{\frac{|V_n|}{|V|}}\right\|_F^2
\end{align*}
Now, we expect the RHS of Equation~\eqref{eq:kmeans_final} to depend only on $\widetilde{a}_n(\bm{X}_n,\bm{W})$, since $\widetilde{b}_n(\cdot)$ no longer depends on $\bm{X}_n$.
If the graphs are also of the same size, then for each graph, the objective function of  \model\ is equivalent to that of \isomodel. Lemma \ref{lemma:same_obj} formalizes this  (See Appendix \ref{subsec:proof_lemma_2} for proof):
\begin{lemma}
Let the pair $(\bm{X},\bm{\Theta})$ represent a joint-SBM with $K$ communities for $N$ graphs. If the sizes of the graphs are equal, i.e., $|V_n| = |V|/N$ and the graphs are balanced in expectation in terms of communities, i.e., assume the same distribution of cluster membership for all graphs: $\bm{X}_{ni} \stackrel{iid}{\sim} \text{Mult}(\bm{\zeta})$ for all $n \in [1,...,N]$ and $i \in [1,...,|V_n|]$. Then, 
\begin{equation}
\mathbb{E}[\eta_n(\bm{X}_n,\bm{W})] \propto \|\bm{X}_{n} \bm{W}_n - \widehat{\bm{U}}_n\|_F^2
\end{equation}
\label{lemma:same_obj}
\end{lemma}
\vspace{-6mm}

Lemma \ref{lemma:same_obj} illustrates the scenario when clustering the nodes of each graph individually (Isolated) and jointly have the same solution (i.e., based on optimizing $\bo{Q}_n$). 
However, this is only true with respect to the node assignments for each individual graph. 
If we look to the global assignments, the Isolated and the Joint model are expected to be the same only up to $N\!-\!1$ permutations of the community labels. 
Thus, for a good global clustering result, the Isolated model needs an extra $\mathcal{O}(NK^2)$ step to realign the 
estimated $\bm{X}_n$s, which can also introduce additional error.
If the data does not have graphs of the same size or the distribution of clusters varies across graphs, then the Isolated model will likely miss some blocks on each graph. 
Our joint method avoids these issues by using pooled information across the graphs to improve estimation. See Appendix \ref{subsec:toydata} for a toy data example. Appendix~\ref{sec:consist} also includes a discussion of consistency.

\section{Related work}
Community detection in graphs has seen a lot of recent attention. We focus on extensions to multiple graphs, emphasizing two relevant directions: heterogeneity across communities and nonaligned graphs.
For the former,~\cite{karrer2011stochastic} propose a degree corrected SBM to account for heterogeneity inside a community, though \cite{gulikers2017spectral} showed that this fails to retrieve true communities in a high heterogeneous setting. \cite{ali2018improved} introduce a normalized Laplacian form that account for high heterogeneous scenarios, however this comes at a high computational cost. 
\cite{mucha2010community} work with multiplex networks for time-dependent data, in their case each edge has multiple layers (attributes) which can be viewed as multiple aligned graphs. Our work is in a different domain, we work with {\em non aligned} networks. 

More generally, methodology for multiple graphs can be divided into geometric ~\cite{ginestet2014,asta2014geometric} and model-based ~\cite{durante2016,duvenaud2015convolutional} approaches. The first approach seeks to characterize graphs topologically and explore hyperspace measures. \cite{ginestet2014} introduces a geometric characterization of the network using the so-called Fr\'echet mean while their approaces are mathematically elegant, they are substantially less flexible than our work. The second approach aims to embed graph to a lower dimension space without oversimplifying the problem by making use of latent models.  \cite{durante2016,gomes2018multiple} proposed a mixture of latent space model to perform hypothesis testing on population of binary and weighted aligned graphs, respectively.
On the other hand, \citep{duvenaud2015convolutional} worked with non-aligned graphs, $\bo{X}_{ni} \ne \bo{X}_{n'i}$, modeling node features conditioned on its neighbors. They provided a convolutional neural network approach which is invariant to node permutation. 
%
However, unlike us, their method assumes knowledge of node features $X_{ni}$.
Overall, the closest work to our model is from \cite{muk2017} where a similarity measure is used to align communities across graphs. In our approach, the model does not rely on a distance to compare graph communities as we cluster nodes across graphs jointly. 

\section{Synthetic experiments}
In this section, we evaluate our algorithm for the 
joint SBM model with synthetic data where the ground truth is known. We 
 generate data using the following generative process: 

\vspace{-8mm}
\begin{equation}
\begin{aligned}
\begin{aligned}[l]
\bm{\pi}_n|\alpha  \sim \text{Dir}\left(K,\frac{1}{\alpha K}\right),\\
\end{aligned}
\hspace{.5cm}
&
\hspace{.15cm}
\begin{aligned}[5]
X_{ni} | \bm{\pi}_n  \sim \text{Mult}(\bm{\pi}_n),
\end{aligned}\\
a_{nij}| \bm{X}_{ni},\bm{X}_{nj},\left(\theta_{kl}\right)_{k=1,l=1}^{K,K}  & \stackrel{ind}{\sim} \text{Bern}(\bm{X}_{ni}\bm{\Theta}.\bm{X}_{nj}^{T}) \\
\end{aligned}
\label{eq:gen_synthetic}
\end{equation}

\noindent For all experiments, we use $K=6$ and the same $\bm{\Theta}$, shown in Figure \ref{fig:syn_theta_dist} (left).  %
We vary hyperparameters including the number of graphs $N$, the individual graph sizes $|V_n|$, and the heterogeneity of clusters sizes in each graph $\alpha$. Recall that $\alpha \approx 0$ corresponds to a homogeneous setting (similar $\bm{\pi}$ across all graphs), and that as $\alpha$ increases the clusters become more heterogeneous. 

We evaluate the performance of our joint spectral clustering algorithm (\jspec), and compare against separately running spectral SBM on the individual graphs and then aligning (\isospec). We also include as a baseline \texttt{node2vec}~\cite{node2vec-kdd2016}, which embeds the nodes into a low-dimensional vector space and clusters the embeddings. 
For \isospec, we used the {\texttt{Blockmodels}} R package~\cite{leger2016blockmodels}, and for \texttt{node2vec} we used a Python implementation~\cite{node2vec-kdd2016}.

First, we assess community retrieval performance and global $\bm{\Theta}$ estimation for each model. We design two sets of experiments, one for each assessment objective: %

\vspace{-4mm}
\begin{enumerate}
	\item Community retrieval $\widehat{\bm{X}}_n$:
\vspace{-2mm}
	\begin{enumerate}
	\item Fixed $\alpha\!=\!1$. $N\!=\![5,50,100,200]$. For each $N$, $|V_n|\!=\![25,50,100,200,500]$.
	Figure \ref{fig:syn_theta_dist}(right) shows the $\pi_n$s, i.e. the proportion of blocks for each graph per scenario;
	\item Fixed $N \!=\!100$ and $|V_n| \!=\! 200$. We generate datasets for varying values of $\alpha \in [0.1,2]$; 
\vspace{-2mm}
\end{enumerate}
    \item Global $\widehat{\bm{\Theta}}$:
\vspace{-2mm}
    \begin{enumerate}
    	\item Fixed $\alpha\!=\!1$. $N\!=\![50,200,400,600,800,1000]$. For each $N$, $|V_n|\!=\![25,50,200,500]$.
    \end{enumerate}
\end{enumerate}
\vspace{-2mm}
\begin{figure}[htbp]
\begin{minipage}{0.47\textwidth}
\begin{figure}[H]
    \centering
    \includegraphics[height=4.5in]{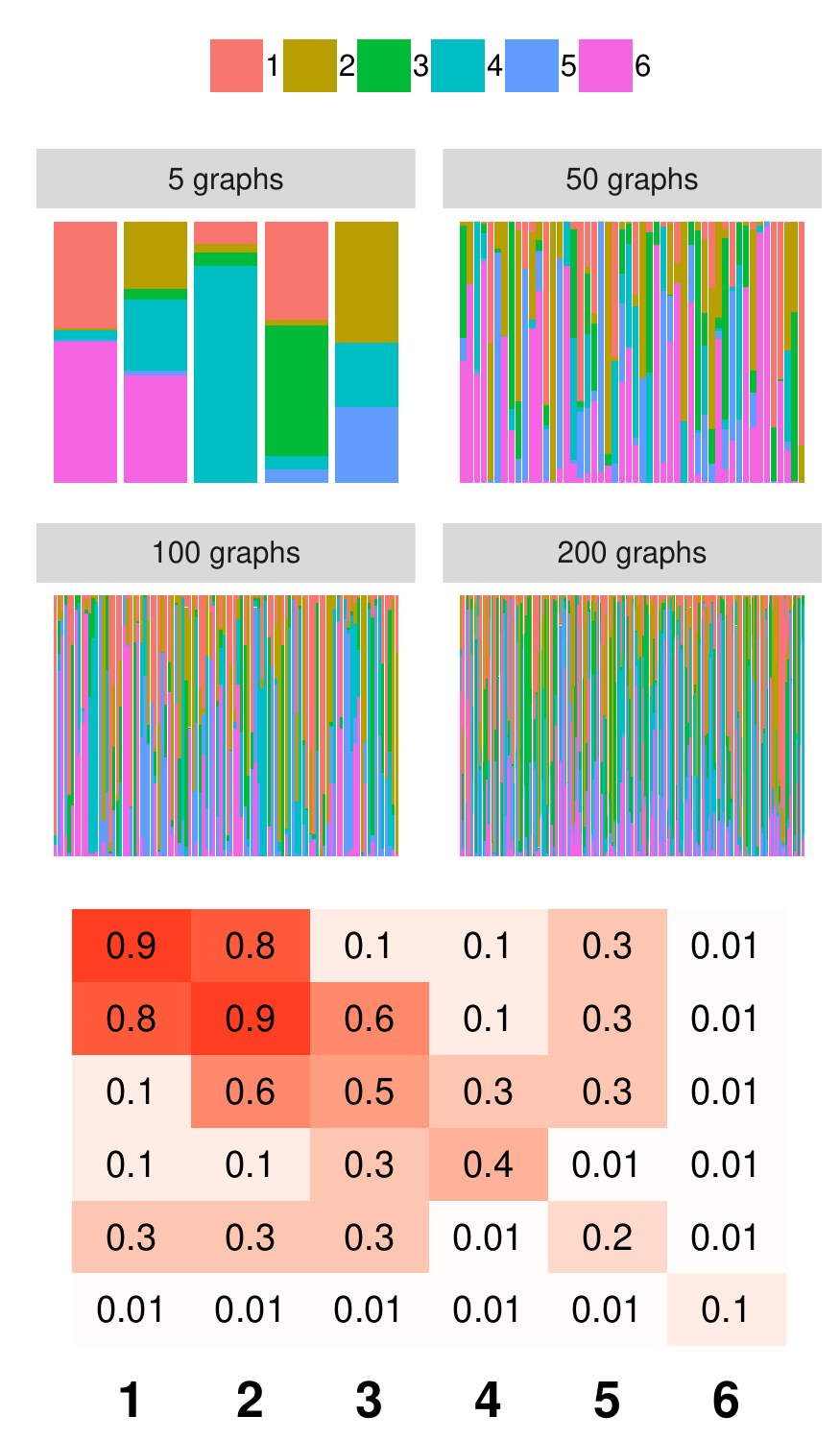}
	\caption{Distribution of blocks for $5$, $50$, $100$, $200$ graphs (top) and synthetic data: True connectivity matrix $\bm{\Theta}$ (bottom).}
	\label{fig:syn_theta_dist}
\end{figure}
\end{minipage}
\begin{minipage}{0.47\textwidth}
\begin{algorithm}[H]
\caption{JointSpec SBM}\label{algo:jointSpec2}
 \begin{algorithmic}
 \STATE \textbf{Input}  $N$ adjacency matrices $\bm{A}_n  \in \{0,1\}^{|V_n| \times |V_n|}$, number of communities $K$ and tolerance  $\varepsilon$
 \FOR{$n \in [1,...,N]$}
 \STATE \textbf{Compute} $\widehat{\bo{U}}_n \in \mathbb{R}^{|V_n| \times K} $ and $\widehat{\bo{D}}_n \in \mathbb{R}^{K \times K}$ as the leading $K$  eigenvectors and eigenvalues of $\bo{A}_n $. \\
\STATE \textbf{set} $\widehat{\bo{Q}}^*_n \gets |\widehat{\bo{U}}_n\widehat{\bo{D}}_n|\sqrt{\frac{|V|}{|V_n|}}$
\STATE \textbf{Initialize} 
$\widehat{\bm{X}}_n$ randomly\\
\STATE \textbf{set}  $\widehat{\Delta}_n^2 \gets \widehat{\bm{X}}_n^T\widehat{\bm{X}}_n$ \\
\ENDFOR
\STATE \textbf{compute}  $\widehat{\Delta}^2 \gets \sum_n^N \widehat{\Delta}_n^2$ \\
\STATE \textbf{set}  $\text{loss}_0 \gets 0$,  $\text{loss}_1 \gets 1$ and $t \gets 1$ \\
 \WHILE{$|\text{loss}_{t} - \text{loss}_{t-1}| > \varepsilon$}
 \STATE \textbf{update}  $t \gets t + 1$ \\
 \FOR{$n \in [1,...,N]$}
 \STATE \textbf{compute}  $\gamma_n \gets \text{tr}(\widehat{\Delta}_n^2\widehat{\Delta}^{-2}) $ \COMMENT{Eq. \eqref{eq:def_gamman}}
\ENDFOR
 \STATE \textbf{update} $\widehat{\bo{W}}$ \COMMENT{Eq. \eqref{eq:w_update_rule}}
 \FOR{$n \in [1,...,N]$}
  \FOR{$i \in [1,...,|V_n|]$}
    \FOR{$k \in [1,...,K]$}
      \STATE \textbf{compute} $\tilde{\Delta}_{ni}(k)$ \COMMENT{Eq. \eqref{eq:def_gammaTilde}}
       \STATE \textbf{compute} $\omega_{ni}(k)$ \COMMENT{Eq. \eqref{eq:def_omega2}}
    \ENDFOR
    \STATE \textbf{update} $\widehat{\bo{X}}_{ni} \gets \text{onehot} (\argmin_k \omega_{ni}(k))$
    \ENDFOR
    \STATE \textbf{compute}  $\widehat{\Delta}_n^2 \gets \widehat{\bm{X}}_n^T\widehat{\bm{X}}_n$ \\
    \STATE \textbf{update}  $\eta_n(\widehat{\bm{X}}_n,\widehat{\bm{W}})$ \COMMENT{Eq. \eqref{eq:opt_kmeans_graphLevel_v3}} \\
\ENDFOR
 \STATE \textbf{compute}  $\widehat{\Delta}^2 \gets \sum_n^N \widehat{\Delta}_n^2$ \\
 \STATE \textbf{compute}  $\text{loss}_t\gets \sum_n^N \eta_n(\widehat{\bm{X}}_n,\widehat{\bm{W}})$
 \ENDWHILE
 \end{algorithmic}
\end{algorithm} 
\end{minipage}
 \end{figure}

We measured performance both quantitatively, using Normalized Mutual Information (NMI) and Standardized Square Error (SSE), and qualitatively, by visualizing the estimated connectivity matrix for each approach. 
For multiple graph datasets, we measure individual graph NMIs as well 
as the overall NMI across all graphs, in each case, comparing the estimated membership matrices $\widehat{\bo{X}}_n$ with the ground truth $\bo{X}_n$. 
To measure the quality of the estimated connectivity matrix $\bm{\Theta}$, we used the standardized square error, the square error normalized by the true variance. Thus, 

\vspace{-5mm}
\begin{equation}
SSE = \sum_{k}^K\sum_{l}^K\frac{(\widehat{\theta}_{kl}-\theta_{kl})^2}{\theta_{kl}(1-\theta_{kl})}
\label{eq:se_var_def}
\end{equation}
\vspace{-5mm}

We normalize the square error by the true variance because very high (or very low) edge probabilities have lower variances and need to be up-weighted accordingly. For good $\bm{\Theta}$ estimates, we expect Eq.~\eqref{eq:se_var_def} to be close to zero. 

\isospec\ and \texttt{node2vec} are by construction nonaligned. %
In order to {align} them, we 
(1) rank each community on each graph based on  $\diag(\bm{\Theta}_n)$, then (2) re-order the connectivity matrix and membership accordingly. 

\vspace{-2mm}
\subsection{Results}

\paragraph{Community retrieval $\widehat{\mathbf{X}}_n$, fixed $\alpha=1$:}
Figure \ref{fig:syn_results1}(left) shows the NMI curves for increasing number of nodes for different numbers of graphs (top row: individual NMIs, and bottom row: overall NMI). For individual NMIs, each data point records the median over the graphs and the shaded region shows the interquartile range. 
We see that \jspec\ outperformed \isospec\ and \texttt{node2vec} for both overall and individual NMI. That the overall NMI for \isospec and \texttt{node2vec} is poor is unsurprising, given the two-stage alignment procedure involved. 
Interestingly however, they perform poorly on the individual NMIs as well. 
This indicates that regardless the choice of realignment, the use of \textit{only local information is not enough to accurately assign nodes to clusters in multiple graph domains}, and that it is important to pool statistical information across graphs. 
\begin{figure}[htbp]
  \centering
  \includegraphics[width=.95\textwidth]{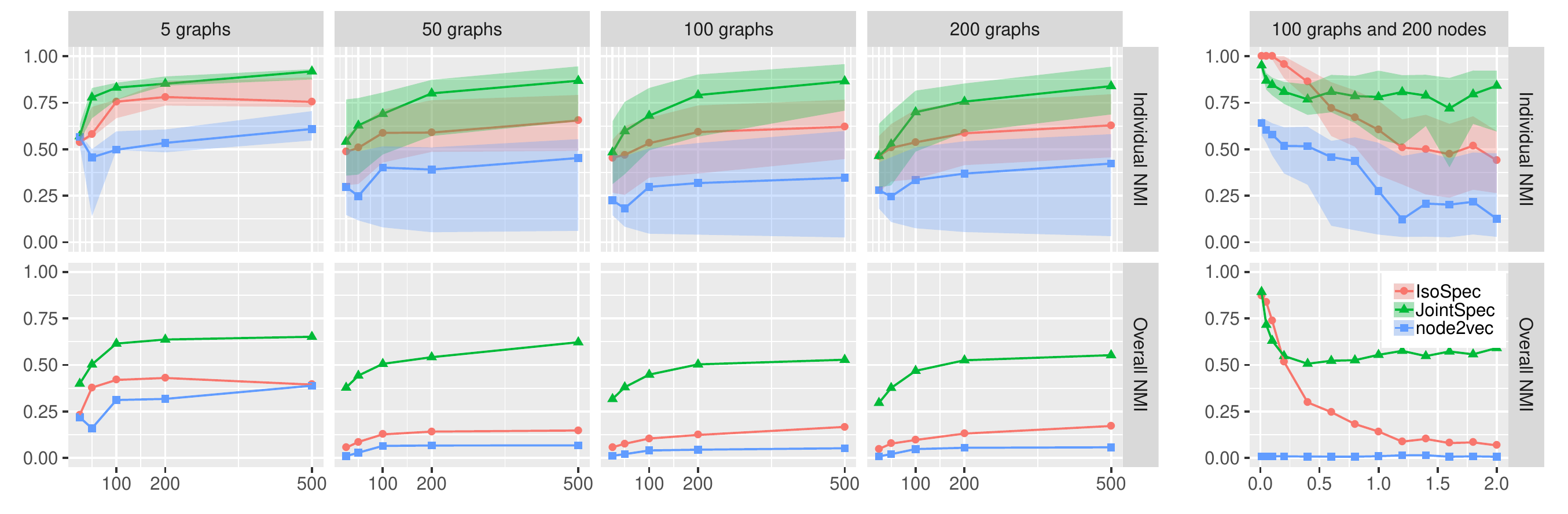}
\vspace{-1mm}
  \caption{\textbf{(Left)} Fixed $\alpha$: NMI for each scenario ($5$, $50$, $100$, $200$ graphs) for increasing number of nodes ($25$, $50$, $100$, $200$ and $500$). \textbf{(Right)} Fixed $N=100$ and $|V_n|=200$: NMI curves for each model for increasing heterogeneity ($\alpha$). Top row: median and the interquartile range curves of the individual NMIs. Bottom row: overall NMI curves.}
  \label{fig:syn_results1}
\end{figure}

In terms of estimating the connectivity matrix, \texttt{node2vec} performs worst and the \jspec\ estimates are the most similar to the true connectivity. Figure \ref{fig:thetas} (top row) shows the estimated connectivity matrix for each approach for the case of $200$ graphs, $500$ nodes per graph.  

\begin{figure}[h]
\begin{minipage}{0.45\textwidth}
\begin{figure}[H]
	\centering
	\includegraphics[width=2.9in]{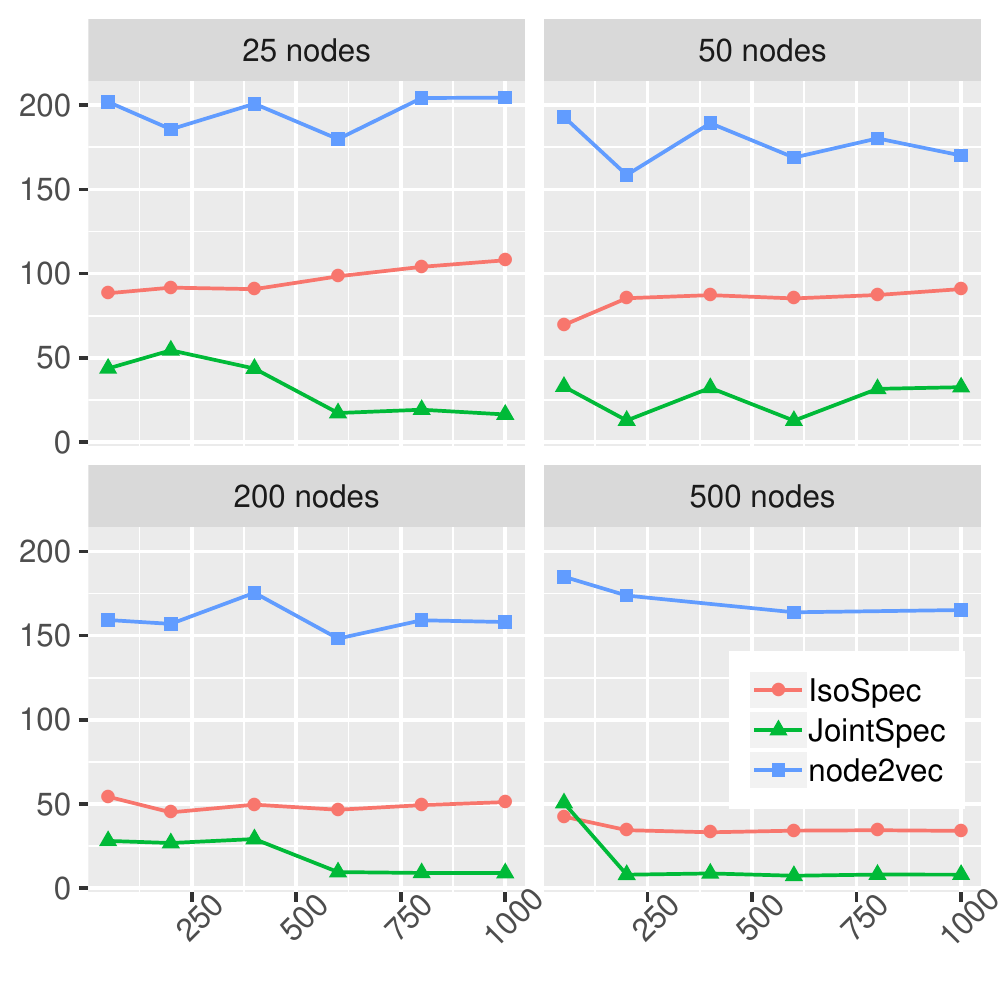}
\vspace{-6mm}
	\caption{
	Standardized square error of $\bm{\Theta}$ (Eq.~\eqref{eq:se_var_def}) for increasing number of graphs $N$, for $|V_n| = 25, 50, 200, 500$.}
	\label{fig:consistency_theta}
\end{figure}
\end{minipage}
\hspace{5mm}
\begin{minipage}{0.45\textwidth}
\begin{figure}[H]
	\centering 
	\includegraphics[width=\textwidth]{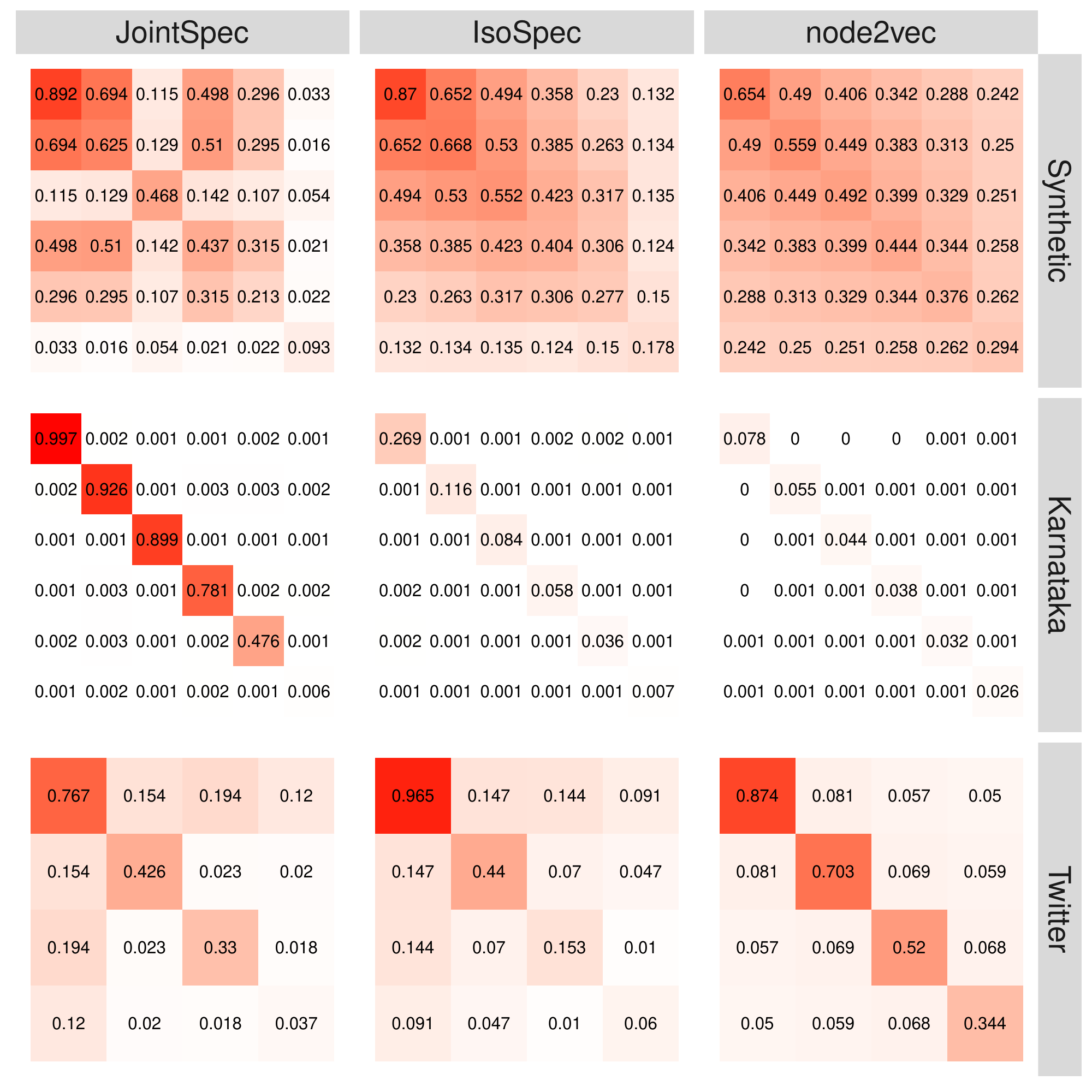}
\vspace{-4mm}
	\caption{Connectivity matrix estimated ($\widehat{\bm{\Theta}}$) by each approach  (columns) and for each dataset (rows) (synthetic data, Karnataka villages, and Twitter). }
	\label{fig:thetas}
\end{figure}
\end{minipage}
\end{figure}
\vspace{-2mm}
\paragraph{Community retrieval $\widehat{\mathbf{X}}_n$, fixed $N\!=\!100$ and $|V_n|\!=\!200$:}
Here, we evaluate cluster retrieval performance as the cluster sizes became more heterogeneous (e.g. $\alpha$ increases). 
Figure \ref{fig:syn_results1}(right) shows the NMI curves. The results %
suggest that \isospec\  and \jspec\ have similar performance when the graphs are balanced in terms of distribution of nodes over communities (low values of $\alpha$). However, the NMI curves diverge as $\alpha$ increases and for more heterogeneous settings (unbalanced graphs), the Joint model outperforms \isospec\ and \texttt{node2vec} by a large amount, both for overall and individual cluster retrieval performance. 

Furthermore, we consider more complex re-alignment procedures, and two additional cluster retrieval performance measures (adjusted Rand index and misclutering rate). Results are in accordance with the ones shown in Figure \ref{fig:syn_results1}(right) in  which \model outperforms the Isolated models. Appendix \ref{subsec:realign_measures} presents those results.

We also investigated settings where the individual graph sizes varied over the dataset. To generate the data we used an overdispersed negative binomial distribution to sample graph sizes. The results are in Appendix \ref{subsec:experiments_3}. Overall, we found that varying $\alpha$ had more impact on cluster retrieval performance than varying the spread of graph sizes. 

\paragraph{Global $\widehat{\mathbf{\Theta}}$:}
Figure \ref{fig:consistency_theta} shows the standardized square error (SSE)~\eqref{eq:se_var_def} for the $\bm{\Theta}$ estimates. %
(lower values mean better estimates of the true $\bm{\Theta}$). Again, \jspec\ outperforms \isospec\ and \texttt{node2vec} in all scenarios. Even when each graph does not have many nodes (e.g. $|V_n|=25$), statistical pooling allows \jspec\ to achieve good performance that is comparable with settings with larger nodes. These results also illustrate the consistency of our estimation scheme, with decreasing error as the number of samples increases. %

\section{Real world experiments}

We consider two datasets: a dataset of Indian villages, and a Twitter 
dataset from the political crisis in Brazil.

\vspace{-4mm}
\paragraph{Karnataka village dataset}
This\footnote{available at https://goo.gl/Vw66H4}  consists of a household census of  $75$ villages in Karnataka, India. %
Each village is a graph, each person is a node, and edges represent if one person went to another person's house or vice versa. Overall, we have $75$ graphs varying in size from $354$ to $1773$ nodes. 

Figure \ref{fig:vil_assignments} (top) shows some demographics proportions for each village: these are highly heterogeneous. Looking at 
Caste, for instance,  we have that most villages consist of ``OBC'' and ``Scheduled Caste'', but some have a very large proportion of people 
identified as ``General''.
Religion and Mother-tongue seem to have same behavior. Overall, %
we expect that Lemma \ref{lemma:same_prop_pop} will not hold for this scenario, and expect our approach to be better suited than \isospec. 

\begin{figure}[h]
\begin{minipage}{.45\textwidth}
\begin{table}[H]
	\caption{Overall and Individual (mean and interquartile range) NMIs for Karnataka villages dataset.}
	\label{tab:nmi_village}
	\centering
	\small
	\begin{tabularx}{\textwidth}{l|cc}
	\hline

	\hline
	\multirow{2}{*}{\textbf{Model}} & \textbf{Overall} & \textbf{Individual}  \\
	& \textbf{NMI}& \textbf{NMI} \\
	\hline
	JointSpec & $\textbf{0.134}$ & $\textbf{0.6126}$ $\textbf{[0.211,1]}$\\
	IsoSpec & $0.091$ & $0.4995$ $[0.191,0.789]$\\
	node2vec & $0.113$ & $0.43807$ $[0.052,0.744]$\\
	\hline
	\end{tabularx}
\end{table}
\begin{figure}[H]
	\centering
	\includegraphics[width=\textwidth]{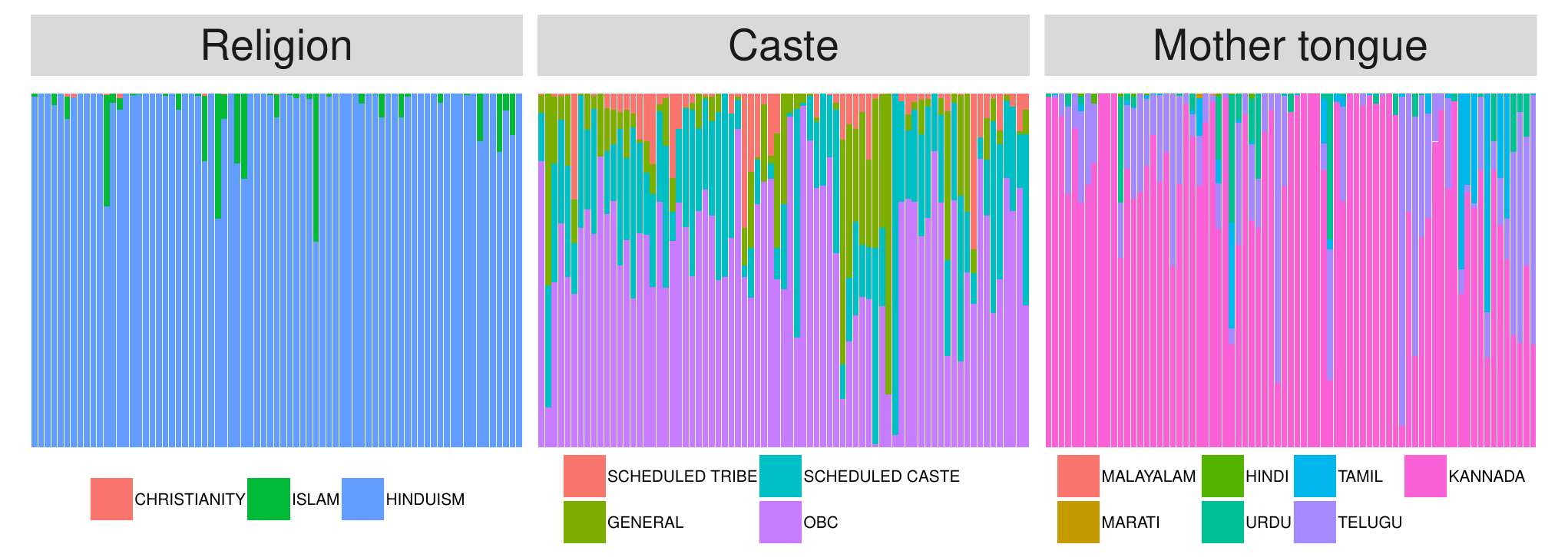}
\vspace{-6mm}
	\caption{Karnataka villages demographics proportions  for Religion, Caste and Mother tongue.}
	\label{fig:vil_assignments}
\end{figure}
\end{minipage}
\hspace{5mm}
\begin{minipage}{.45\textwidth}
\begin{table}[H]
	\caption{Top $3$ words assigned to communities by each model and each side (pro and against government)}
	\label{tab:words}
	\centering
	\scriptsize
	\begin{tabular}{l|l|l|l|l}
	\hline
	\hline
	\multirow{2}{*}{$\bm{K}$} & \multicolumn{2}{c}{\textbf{JointSpec}}  & \multicolumn{2}{c}{\textbf{IsoSpec}} \\
	& \color{red}{Pro} & \color{blue}{Against} & \color{red}{Pro} & \color{blue}{Against}\\
	\hline
		\multirow{3}{*}{1} & \color{red}{naovaitergolpe} & \color{blue}{foradilma}& \color{red}{naovaitergolpe} & \color{blue}{foradilma}\\
   						   &          \color{red}{golpe} & \color{blue}{forapt}&           \color{red}{golpe}&    \color{blue}{forapt}\\
    					   &      \color{red}{dilmafica} & dilma &       \color{blue}{foradilma}&     dilma\\
    	\hline
		\multirow{3}{*}{2} &       hora		& janaiva		&       turno  &         arte   \\
   						   &     galera		& compartilhar			&       venceu &      objetos   \\
    					   & democralica	& faltam		&  anavilarino &    apropriou   \\
    	\hline
		\multirow{3}{*}{3} &      sociais	&        mito	&      rua  &        rua \\
   						   &   \color{red}{coxinha}	&\color{blue}{elite}	&     april &     \color{red}{coxinha} \\
    					   &     naonaors	&compartilhar	& continuar &      \color{blue}{elite} \\    
    	\hline
		\multirow{3}{*}{4} &    vai			& \color{blue}{lulanacadeia}	&      \color{red}{dilmafica} &impeachment \\
   						   & brasil			& lula			&      dilma&       brasil\\
    					   &   povo			& vai 			& \color{blue}{forapt}& \color{blue}{lulanacadeia}\\
	\hline

	\hline
	\end{tabular}
\vspace{-2mm}
\end{table}
\end{minipage}
\end{figure}

We consider a conservative setting of six communities ($K = 6$) for all models. Here, ground truth is unknown and we cannot assess the performance of the models precisely. 
To quantitatively measure performance, we compare the demographic variables presented earlier (religion, caste and mother-tongue) with the communities assigned using each model. %
Table \ref{tab:nmi_village} shows the NMIs computed in this setting, we can see that \jspec\ outperformed the baselines on both the overall NMI and the Individual NMI. 
%
%
We see in Figure \ref{fig:thetas} (middle row) that the connectivity matrices estimated using IsoSpec and node2vec are noisy and weakly structured, making interpretation difficult. The connectivity matrix that our method estimated presents a very strong structure, where the within-community probability is very large and across-community is very low.

%
%
%
%
%
%
%

%
%
%
%
%
%

%
%
%
%
%
%

%

\vspace{-2mm}
\paragraph{Twitter}
We crawled Twitter from April 6th to May 31st 2016 to 
collect tweets  from hashtags of $7382$ users from both sides of the political crisis in 
Brazil. 
One side was for the impeachment of the former president, Dilma Rousseff, with the opposition describing the process as a coup. 
We constructed graphs based on the word co-occurrence matrices forming edge-links between words which were co-tweeted more than 20\% of the time for a given user. Thus, networks represent users and nodes represent words.
The resulting networks are somewhat homogeneous after controlling for the side the user is in. The sizes of the graphs vary from $25$ to $1039$ nodes. 

We used four communities ($K = 4$) in this experiment, and Figure \ref{fig:thetas} (bottom row) shows the estimated connectivity matrices. Now, on the surface, both \jspec\ and \isospec\ have similar behavior, with \texttt{node2vec} estimating an extra cluster.
Despite this superficial similarity, the words assigned to communities by the two models differ significantly. 
In order to assess words for each community, we split users between the two sides (pro and against government), and assign each word to its most frequent community across graphs. Table \ref{tab:words} shows the top three words per community per model. We colored words based on whether it is a more pro government (red), against it (blue) or neutral (black). For \jspec, we see that community 1 has important key arguments per side such as ``naovaitergolpe'' (no coup)  and ``foradilma''(resign dilma). We also see that community 2 seems to have some stopwords such as ``hora''(time) and ``compartilhar'' (share), and community 3 consist of aggressive and pejorative terms each side uses against the other like ``elite'' and ``coxinha''. For \isospec, no such inferences are made, and words do not seem to reflect a clear pattern. \texttt{node2vec} had the least informative clustering performance, where the top words for the different communities were all neutral. %

Since there is partial alignment across the Twitter graphs (i.e., words can appear as nodes in multiple graphs), we can assess the consistency of the clustering approaches by evaluating the entropy of community assignments across the graphs (per word). Although a word can be used in multiple contexts, we expect that a ``good'' clustering will assign words mostly to the same clusters, therefore resulting in a lower entropy. The results are in Appendix \ref{subsec:experiments_4}. Overall, we found that \jspec\ had the lowest entropy. 

%

%
%
%
%
%
%
%
%
%
%
%
%
%
%
%
%
%
%
%
%

%
%
%
%
%
%
%
%
%
%
%
%

%
%
%
%
%
%
%
%
%
%
%
%
%

%
%
%
%

%
%

%
%
%
%
%
%

%

%
%
%
%
%
%
%
%
%
%
%
%

%
%

%
%
%
%
%

%
%
%
%
%
%
%
%
%
%
%

%
%
%
%
%
%
%
%
%
%
%
%
%
%
%
%
%
 %
 %
%
%

%

%
%
%
%

%

%
\section{Conclusion}

In this work, we consider the problem of multiple graph community detection in a heterogeneous setting. We showed that we need to jointly perform 
stochastic block model decompositions in order to be able to estimate a reliable global structure. We compared our methods with a two-step approach we 
called the \isomodel, and \texttt{node2vec}. Our method outperformed the baselines on global measures (overall NMI and SSE of the connectivity matrices), but interestingly also on local measures (individual NMI). This demonstrates that our method is more accurately able to assign nodes to clusters regardless of the choice of re-alignment procedure. 
Overall, the \isomodel \ does not pool global information in the inference step which indicates that it can only be used in highly homogeneous scenarios. 

\bibliography{mybib}

\begin{thebibliography}{39}
\providecommand{\natexlab}[1]{#1}
\providecommand{\url}[1]{\texttt{#1}}
\expandafter\ifx\csname urlstyle\endcsname\relax
  \providecommand{\doi}[1]{doi: #1}\else
  \providecommand{\doi}{doi: \begingroup \urlstyle{rm}\Url}\fi

\bibitem[Liben-Nowell and Kleinberg(2007)]{liben2007link}
David Liben-Nowell and Jon Kleinberg.
\newblock The link-prediction problem for social networks.
\newblock \emph{journal of the Association for Information Science and
  Technology}, 58\penalty0 (7):\penalty0 1019--1031, 2007.

\bibitem[Wang et~al.(2015)Wang, Xu, Wu, and Zhou]{wang2015link}
Peng Wang, BaoWen Xu, YuRong Wu, and XiaoYu Zhou.
\newblock Link prediction in social networks: the state-of-the-art.
\newblock \emph{Science China Information Sciences}, 58\penalty0 (1):\penalty0
  1--38, 2015.

\bibitem[Sarwar et~al.(2001)Sarwar, Karypis, Konstan, and
  Riedl]{sarwar2001item}
Badrul Sarwar, George Karypis, Joseph Konstan, and John Riedl.
\newblock Item-based collaborative filtering recommendation algorithms.
\newblock In \emph{Proceedings of the 10th international conference on World
  Wide Web}, pages 285--295. ACM, 2001.

\bibitem[Linden et~al.(2003)Linden, Smith, and York]{linden2003amazon}
Greg Linden, Brent Smith, and Jeremy York.
\newblock Amazon. com recommendations: Item-to-item collaborative filtering.
\newblock \emph{IEEE Internet computing}, 7\penalty0 (1):\penalty0 76--80,
  2003.

\bibitem[Koren and Bell(2015)]{koren2015advances}
Yehuda Koren and Robert Bell.
\newblock Advances in collaborative filtering.
\newblock In \emph{Recommender systems handbook}, pages 77--118. Springer,
  2015.

\bibitem[Fortunato(2010)]{fortunato2010community}
Santo Fortunato.
\newblock Community detection in graphs.
\newblock \emph{Physics reports}, 486\penalty0 (3-5):\penalty0 75--174, 2010.

\bibitem[Aicher et~al.(2014)Aicher, Jacobs, and Clauset]{aicher2014learning}
Christopher Aicher, Abigail~Z Jacobs, and Aaron Clauset.
\newblock Learning latent block structure in weighted networks.
\newblock \emph{Journal of Complex Networks}, 3\penalty0 (2):\penalty0
  221--248, 2014.

\bibitem[Newman(2016)]{newman2016community}
MEJ Newman.
\newblock Community detection in networks: Modularity optimization and maximum
  likelihood are equivalent.
\newblock \emph{arXiv preprint arXiv:1606.02319}, 2016.

\bibitem[Cozzo et~al.(2013)Cozzo, Kivel{\"a}, De~Domenico, Sol{\'e}, Arenas,
  G{\'o}mez, Porter, and Moreno]{cozzo2013clustering}
Emanuele Cozzo, Mikko Kivel{\"a}, Manlio De~Domenico, Albert Sol{\'e}, Alex
  Arenas, Sergio G{\'o}mez, Mason~A Porter, and Yamir Moreno.
\newblock Clustering coefficients in multiplex networks.
\newblock \emph{arXiv preprint arXiv:1307.6780}, 2013.

\bibitem[Mucha et~al.(2010)Mucha, Richardson, Macon, Porter, and
  Onnela]{mucha2010community}
Peter~J Mucha, Thomas Richardson, Kevin Macon, Mason~A Porter, and Jukka-Pekka
  Onnela.
\newblock Community structure in time-dependent, multiscale, and multiplex
  networks.
\newblock \emph{science}, 328\penalty0 (5980):\penalty0 876--878, 2010.

\bibitem[Zhou et~al.(2017)Zhou, Li, Li, Zhang, and Cui]{zhou2017graph}
HongFang Zhou, Jin Li, JunHuai Li, FaCun Zhang, and YingAn Cui.
\newblock A graph clustering method for community detection in complex
  networks.
\newblock \emph{Physica A: Statistical Mechanics and Its Applications},
  469:\penalty0 551--562, 2017.

\bibitem[Von~Landesberger et~al.(2016)Von~Landesberger, Brodkorb, Roskosch,
  Andrienko, Andrienko, and Kerren]{von2016mobilitygraphs}
Tatiana Von~Landesberger, Felix Brodkorb, Philipp Roskosch, Natalia Andrienko,
  Gennady Andrienko, and Andreas Kerren.
\newblock Mobilitygraphs: Visual analysis of mass mobility dynamics via
  spatio-temporal graphs and clustering.
\newblock \emph{IEEE transactions on visualization and computer graphics},
  22\penalty0 (1):\penalty0 11--20, 2016.

\bibitem[Peixoto(2015)]{peixoto2015inferring}
Tiago~P Peixoto.
\newblock Inferring the mesoscale structure of layered, edge-valued, and
  time-varying networks.
\newblock \emph{Physical Review E}, 92\penalty0 (4):\penalty0 042807, 2015.

\bibitem[Sarkar et~al.(2012)Sarkar, Chakrabarti, and
  Jordan]{sarkar2012nonparametric}
Purnamrita Sarkar, Deepayan Chakrabarti, and Michael~I Jordan.
\newblock Nonparametric link prediction in dynamic networks.
\newblock In \emph{Proceedings of the 29th International Coference on
  International Conference on Machine Learning}, pages 1897--1904. Omnipress,
  2012.

\bibitem[Charlin et~al.(2015)Charlin, Ranganath, McInerney, and
  Blei]{charlin2015dynamic}
Laurent Charlin, Rajesh Ranganath, James McInerney, and David~M Blei.
\newblock Dynamic poisson factorization.
\newblock In \emph{Proceedings of the 9th ACM Conference on Recommender
  Systems}, pages 155--162. ACM, 2015.

\bibitem[Durante and Dunson(2014)]{durante2014nonparametric}
Daniele Durante and David~B Dunson.
\newblock Nonparametric bayes dynamic modelling of relational data.
\newblock \emph{Biometrika}, pages 1--16, 2014.

\bibitem[Ginestet et~al.(2017)Ginestet, Li, Balachandran, Rosenberg, Kolaczyk,
  et~al.]{ginestet2014}
Cedric~E Ginestet, Jun Li, Prakash Balachandran, Steven Rosenberg, Eric~D
  Kolaczyk, et~al.
\newblock Hypothesis testing for network data in functional neuroimaging.
\newblock \emph{The Annals of Applied Statistics}, 11\penalty0 (2):\penalty0
  725--750, 2017.

\bibitem[Asta and Shalizi(2015)]{asta2014geometric}
Dena~Marie Asta and Cosma~Rohilla Shalizi.
\newblock Geometric network comparisons.
\newblock In \emph{Proceedings of the Thirty-First Conference on Uncertainty in
  Artificial Intelligence}, pages 102--110. AUAI Press, 2015.

\bibitem[Durante et~al.(2015)Durante, Dunson, and Vogelstein]{durante2015}
Daniele Durante, David~B Dunson, and Joshua~T Vogelstein.
\newblock Nonparametric bayes modeling of populations of networks.
\newblock \emph{Journal of the American Statistical Association}, 2015.

\bibitem[Durante et~al.(2018)Durante, Dunson, et~al.]{durante2016}
Daniele Durante, David~B Dunson, et~al.
\newblock Bayesian inference and testing of group differences in brain
  networks.
\newblock \emph{Bayesian Analysis}, 13\penalty0 (1):\penalty0 29--58, 2018.

\bibitem[Rohe et~al.(2011)Rohe, Chatterjee, and Yu]{rohe2011spectral}
Karl Rohe, Sourav Chatterjee, and Bin Yu.
\newblock Spectral clustering and the high-dimensional stochastic blockmodel.
\newblock \emph{The Annals of Statistics}, pages 1878--1915, 2011.

\bibitem[Lei et~al.(2015)Lei, Rinaldo, et~al.]{lei2015consistency}
Jing Lei, Alessandro Rinaldo, et~al.
\newblock Consistency of spectral clustering in stochastic block models.
\newblock \emph{The Annals of Statistics}, 43\penalty0 (1):\penalty0 215--237,
  2015.

\bibitem[Sarkar et~al.(2015)Sarkar, Bickel, et~al.]{sarkar2015role}
Purnamrita Sarkar, Peter~J Bickel, et~al.
\newblock Role of normalization in spectral clustering for stochastic
  blockmodels.
\newblock \emph{The Annals of Statistics}, 43\penalty0 (3):\penalty0 962--990,
  2015.

\bibitem[Airoldi et~al.(2008)Airoldi, Blei, Fienberg, and
  Xing]{airoldi2008mixed}
Edoardo~M Airoldi, David~M Blei, Stephen~E Fienberg, and Eric~P Xing.
\newblock Mixed membership stochastic blockmodels.
\newblock \emph{Journal of Machine Learning Research}, 9\penalty0
  (Sep):\penalty0 1981--2014, 2008.

\bibitem[Karrer and Newman(2011)]{karrer2011stochastic}
Brian Karrer and Mark~EJ Newman.
\newblock Stochastic blockmodels and community structure in networks.
\newblock \emph{Physical review E}, 83\penalty0 (1):\penalty0 016107, 2011.

\bibitem[Holland et~al.(1983)Holland, Laskey, and
  Leinhardt]{holland1983stochastic}
P.~W. Holland, K.~B. Laskey, and S.~Leinhardt.
\newblock Stochastic blockmodels: Some first steps.
\newblock \emph{Social Networks}, 5:\penalty0 109--137, 1983.

\bibitem[Wasserman and Anderson(1987)]{wasserman1987stochastic}
Stanley Wasserman and Carolyn Anderson.
\newblock Stochastic a posteriori blockmodels: Construction and assessment.
\newblock \emph{Social networks}, 9\penalty0 (1):\penalty0 1--36, 1987.

\bibitem[Erdos and Renyi(1959)]{erdds1959random}
P~Erdos and A~Renyi.
\newblock On random graphs i.
\newblock \emph{Publ. Math. Debrecen}, 6:\penalty0 290--297, 1959.

\bibitem[von Luxburg(2007)]{von2007tutorial}
U.~von Luxburg.
\newblock {A tutorial on spectral clustering}.
\newblock \emph{Statistics and Computing}, 17\penalty0 (4):\penalty0 395--416,
  2007.

\bibitem[Pan and Chen(1999)]{pan1999complexity}
Victor~Y Pan and Zhao~Q Chen.
\newblock The complexity of the matrix eigenproblem.
\newblock In \emph{Proceedings of the thirty-first annual ACM symposium on
  Theory of computing}, pages 507--516. ACM, 1999.

\bibitem[Golub and Van~Loan(2012)]{golub2012matrix}
Gene~H Golub and Charles~F Van~Loan.
\newblock \emph{Matrix computations}, volume~3.
\newblock JHU Press, 2012.

\bibitem[Gulikers et~al.(2017)Gulikers, Lelarge, and
  Massouli{\'e}]{gulikers2017spectral}
Lennart Gulikers, Marc Lelarge, and Laurent Massouli{\'e}.
\newblock A spectral method for community detection in moderately sparse
  degree-corrected stochastic block models.
\newblock \emph{Advances in Applied Probability}, 49\penalty0 (3):\penalty0
  686--721, 2017.

\bibitem[Ali and Couillet(2017)]{ali2018improved}
Hafiz~Tiomoko Ali and Romain Couillet.
\newblock Improved spectral community detection in large heterogeneous
  networks.
\newblock \emph{The Journal of Machine Learning Research}, 18\penalty0
  (1):\penalty0 8344--8392, 2017.

\bibitem[Duvenaud et~al.(2015)Duvenaud, Maclaurin, Iparraguirre, Bombarell,
  Hirzel, Aspuru-Guzik, and Adams]{duvenaud2015convolutional}
David~K Duvenaud, Dougal Maclaurin, Jorge Iparraguirre, Rafael Bombarell,
  Timothy Hirzel, Al{\'a}n Aspuru-Guzik, and Ryan~P Adams.
\newblock Convolutional networks on graphs for learning molecular fingerprints.
\newblock In \emph{Advances in neural information processing systems}, pages
  2224--2232, 2015.

\bibitem[Gomes et~al.(2018)Gomes, Rao, and Neville]{gomes2018multiple}
Guilherme Gomes, Vinayak Rao, and Jennifer Neville.
\newblock Multi-level hypothesis testing for populations of heterogeneous
  networks.
\newblock In \emph{Data Mining (ICDM), 2018 IEEE 18th International Conference
  on Data Mining}. IEEE, 2018.

\bibitem[Mukherjee et~al.(2017)Mukherjee, Sarkar, and Lin]{muk2017}
Soumendu~Sundar Mukherjee, Purnamrita Sarkar, and Lizhen Lin.
\newblock On clustering network-valued data.
\newblock pages 7074--7084, 2017.

\bibitem[Grover and Leskovec(2016)]{node2vec-kdd2016}
Aditya Grover and Jure Leskovec.
\newblock node2vec: Scalable feature learning for networks.
\newblock In \emph{Proceedings of the 22nd ACM SIGKDD International Conference
  on Knowledge Discovery and Data Mining}, 2016.

\bibitem[Leger(2016)]{leger2016blockmodels}
Jean-Benoist Leger.
\newblock Blockmodels: A r-package for estimating in latent block model and
  stochastic block model, with various probability functions, with or without
  covariates.
\newblock \emph{arXiv preprint arXiv:1602.07587}, 2016.

\bibitem[Rand(1971)]{rand1971objective}
William~M Rand.
\newblock Objective criteria for the evaluation of clustering methods.
\newblock \emph{Journal of the American Statistical association}, 66\penalty0
  (336):\penalty0 846--850, 1971.

\end{thebibliography}
\newpage
\clearpage

\newpage
\section{Appendix}
\subsection{Notations}
\vspace{-2mm}
\begin{table}[h!]
\centering
  \begin{tabularx}{.5\textwidth}{c|l}
  \hline
    $K$ & number of communities \\
    $\bm{\Theta}$ & $K\times K$ connectivity matrix \\
    $V_n$ & Nodes in graph $n$ \\
    $V$ & $\bigcup_{n=1}^N V_{n}$\\
    $G_{nk}$ & Nodes in graph $n$ and in community $k$\\
    $G_{\cdot k}$ & $\bigcup_{n=1}^N G_{nk}$\\
    $\bo{A}_n$ & $|V_n| \times |V_n|$ adjacency matrix of graph $n$ \\
    $\bo{P}_n$ & $|V_n| \times |V_n|$  edge probabilities matrix of graph $n$\\
    $\bo{X}_n$ & $|V_n| \times K$ membership matrix of graph $n$\\
    $\Delta_n$ & $\left(\bo{X}_n'\bo{X}_n\right)^{1/2}$\\
    $\bo{U}_n$ & $|V_n| \times K$ matrix of eigenvectors of $P_n$ \\
    $\bo{A}$ & $|V| \times |V|$  block diagonal matrix of all  $\bm{A}_n$\\
    $\bo{P}$ & $|V| \times |V|$ matrix of all edge probabilities \\
    $\bo{X}$ & $|V| \times K$ of all membership matrices stacked \\
    $\Delta$ & $\left(\bo{X}'\bo{X}\right)^{1/2}$\\
    $\bo{U}$ & $|V| \times K$ matrix of $\bo{U}_n$s stacked \\
    $\|\cdot\|$ & Euclidean (vector) and spectral (matrix) norm  \\
    $\|M\|_{F}$ & Frobenius norm of matrix $M$ \\
    $\|M\|_{0}$ & The $l_0$-norm of matrix $M$ \\
  \hline
  \end{tabularx}
  \caption{Notation}
  \label{tab:notation}
\end{table}

\subsection{Derivation of Eq.\eqref{eq:theta_estimator_n} and unbiasedness}
\label{subsec:theta_n_details}
Recall that the connectivity matrix $\bm{\Theta}$ consists of edge probability for within and between communities. Now, say $\theta_{kl}$ is the element of $\bm{\Theta}$ on the $k$th row $l$-th column. Thus, one can estimate $\theta_{kl}$ by counting the number of edges between communities $k$ and $l$ and dividing by the total number of possible edges. For $k\ne l$, the total number of possible edges is the number of nodes in $k$ multiplied by the number of nodes in $l$. For a adjacency matrix $\bm{A}_n$, we can generalize the estimation of the connectivity matrix using matrix notation as
\begin{equation}
\widehat{\bm{S}}_{n} = [\widehat{\bm{X}}_n^{T}\widehat{\bm{X}}_n]^{-1}\widehat{\bm{X}}_n^T\bm{A}_n\widehat{\bm{X}}_n[\widehat{\bm{X}}_n^{T}\widehat{\bm{X}}_n]^{-1}
\label{eq:theta_wrong}
\end{equation}
If assume we know the true membership matrix (i.e. $\widehat{\bm{X}} = \bm{X}_n$), the off-diagonal elements of $\widehat{\bm{S}}_{n}$ in Eq. ~\eqref{eq:theta_wrong} above have unbiased estimates, however the diagonal elements (i.e. the within community probability) are biased. More specifically, the total possible number of edges for nodes in the same communities is being assumed to have self loops which is incorrect in this setting, and also the term $\widehat{\bm{X}}_n^T\bm{A}_n\widehat{\bm{X}}_n$ is double counting the edges within communities. Formally,
\begin{align}
\mathbb{E}\left[\widehat{\bm{S}}_n\right]& = \Delta_n^{-2}\left(\bm{X}_n^T\bm{P}_n\bm{X}_n-\bm{X}_n^T\diag\left(\bm{P}_n\right)\bm{X}_n\right)\Delta_n^{-2} \nonumber\\
& = \bm{\Theta} - \Delta_n^{-2}\diag\left(\bm{\Theta}\right)
\label{eq:unbiased}
\end{align}
Here, we are using the fact that $\mathbb{E}[\bm{A}_n] = \bm{P}_n - \diag\left(\bm{P}_n\right)$. 
Furthermore, we can construct an unbiased estimator for $\bm{\Theta}$ by adding the following term to each diagonal element of $\widehat{\bm{S}}_n$: 
\begin{equation}
\frac{\text{number of edges in cluster $k$}}{|G_{nk}|{|G_{nk}| \choose 2}}
\label{eq:add_term}
\end{equation}
where $|G_{nk}|$ is the number of nodes in community $k$. For all $k$, we have Eq.\eqref{eq:add_term} in matrix notation as
\begin{equation}
\Delta_n^{-2}\left[\mathbb{I}_k-\Delta_n^{-2}\right]^{-1}\diag\left(\widehat{\bm{S}}_n\right)
\label{eq:matrix_term}
\end{equation}
Now, it follows from Eq.~\eqref{eq:unbiased} that
\begin{align}
\mathbb{E}\left[\diag\left(\widehat{\bm{S}}_n\right)\right]&= \left[\mathbb{I}  - \Delta_n^{-2}\right]\diag(\bm{\Theta}) %
\label{eq:expected_term}
\end{align}
Using Eqs.\eqref{eq:theta_wrong} and \eqref{eq:matrix_term}, we get the expression in Eq.\eqref{eq:theta_estimator_n}. And using Eq.~\eqref{eq:unbiased} and ~\eqref{eq:expected_term}, we have
\begin{equation}\mathbb{E}\left[\widehat{\bm{\Theta}}_n\right] = \bm{\Theta}\label{eq:expected_estimator_n}
\end{equation}
We also have that
\begin{equation}
  \text{Var}\left(\widehat{\bm{\Theta}}_n\right)= 
\begin{cases}
    \frac{\theta_{kl}(1-\theta_{kl})}{|G_{nk}||G_{nl}|},& \text{off-diagonal elements}\\
    \frac{\theta_{kk}(1-\theta_{kk})}{{|G_{nk}|\choose 2}},& \text{diagonal elements}
\end{cases}
\label{eq:theta_estimator_n_variance}
\end{equation}
where $|G_{nk}|$ is the number of nodes of graph $n$ in cluster $k$.

\subsection{Proof of Lemma~\ref{lem:kmean_opt} } %
\label{subsec:opt_kmeans_graphLevel_details_v2}

\textit{\textbf{Proof}}
From Equation~\eqref{eq:global_local_expression}, and since $\bo{X}_n = \bo{X}_{n*}$, we have: 
\allowdisplaybreaks{
\begin{align*}
0&= \bm{X}_{n}\bo{W} - \bo{Q}_n\bm{Z}_n^{T}\Delta_n^{-1}\Delta\bo{Z} \\
&= \bm{X}_{n}\bo{W}\bo{Z}^T \Delta^{-1} \Delta_n \bm{Z}_n - \bo{Q}_n \\
&= \bm{X}_{n}\bo{W}\bo{Z}^T \Delta^{-1} \Delta_n \bm{Z}_n - \bo{Q}_n \sqrt{\frac{|V|}{|V_n|}} \sqrt{\frac{|V_n|}{|V|}} \mathbb{I}_K\\
&= \bm{X}_{n}\bo{W}\bo{Z}^T \Delta^{-1} \Delta_n \bm{Z}_n - \bo{Q}_n^* \sqrt{\frac{|V_n|}{|V|}} \mathbb{I}_K \\
&= \bm{X}_{n}\bo{W}\bo{Z}^T \Delta^{-1} \Delta_n \bm{Z}_n - \bo{Q}_n^* \bo{Z}^T \Delta^{-1} \Delta_n \bm{Z}_n + \bo{Q}_n^* \bo{Z}^T \Delta^{-1} \Delta_n \bm{Z}_n - \bo{Q}_n^* \sqrt{\frac{|V_n|}{|V|}} \mathbb{I}_K\\
&= \left(\bm{X}_{n}\bo{W} - \bo{Q}_n^* \right) \bo{Z}^T \Delta^{-1} \Delta_n \bm{Z}_n+ \bo{Q}_n^* \left(\bo{Z}^T \Delta^{-1} \Delta_n \bm{Z}_n - \sqrt{\frac{|V_n|}{|V|}} \mathbb{I}_K \right)
\end{align*}
}

\vspace{-4mm}
\noindent where $\bo{Q}_n^* := \bo{Q}_n\sqrt{\frac{|V|}{|V_n|}}$. {Recall that $\bo{Q}_n$ corresponds to the data from a single graph. $\bo{Q}_n^*$ is then a weighted version, based on the relative number of nodes in the graph.}

From this we can transform Eq.~\eqref{eq:opt_kmeans_graphLevel} to 
\vspace{-3mm}
\begin{equation*}
  \hspace{-2mm}\argmin_{\substack{\bm{X} \in \mathcal{M}_{|V|,K} \\ \bm{W} \in \mathbb{R}^{K \times K}}}  \sum_{n=1}^N \left\|a_n(\bm{X}_n,\bm{W})  + b_n(\bm{X}_n) \right\|^2_F.
\end{equation*}

\vspace{-3mm}
\noindent where

\vspace{-7mm}
$$
\begin{aligned}
a_n(\bm{X}_n,\bm{W}) &:= \left( \bm{X}_{n}\bo{W} - \bo{Q}_n^* \right)\bm{Z}^{T}\Delta^{-1}\Delta_n\bm{Z}_n\\
b_n(\bm{X}_n) &:= \bo{Q}_n^*\left(\bm{Z}^{T}\Delta^{-1}\Delta_n\bo{Z}_n - \sqrt{\frac{|V_n|}{|V|}}\mathbb{I}_K\right) \\
\end{aligned}
$$
\vspace{-3mm}
$\hfill  \qquad \qquad \qquad \qquad \qed $

\subsection{Proof of Lemma~\ref{lem:kmean_opt2}}
\label{subsec:opt_kmeans_graphLevel_details}

\textit{\textbf{Proof}}
\begin{align}
 \frac{1}{2}\left\|a_n(\bm{X}_n,\bm{W})  + b_n(\bm{X}_n) \right\|^2_F &\le  \|a_n(\bm{X}_n,\bm{W})\|_F^2 +  \|b_n(\bm{X}_n)\|^2_F \nonumber \\
& \le \left\| \bm{X}_{n}\bo{W} \!\!-\! \bo{Q}_n^* \right\|_F^2\left\|\bm{Z}^{T}\Delta^{-1}\Delta_n\bm{Z}_n\right\|_F^2 +  \|b_n(\bm{X}_n)\|^2_F  \nonumber \\
&= \left\| \bm{X}_{n}\bo{W} \!\!-\! \bo{Q}_n^* \right\|_F^2 \gamma_n + \left\|\bo{Q}_n^*\bm{Z}^T\Delta^{-1}\Delta_n\bm{Z}_n - \bo{Q}_n^*\sqrt{\frac{|V_n|}{|V|}} \right\|_F^2 \nonumber \\
&\le \left\| \bm{X}_{n}\bo{W} \!\!-\! \bo{Q}_n^* \right\|_F^2 \gamma_n + \left\| \left|\bo{Q}_n^*\right|\Delta^{-1}\Delta_n + \left|\bo{Q}_n^*\right| \sqrt{\frac{|V_n|}{|V|}} \right\|_F^2 \nonumber \\
&= \left\| \bm{X}_{n}\bo{W} \!\!-\! \bo{Q}_n^* \right\|_F^2 \gamma_n + \left\| \left|\bo{Q}_n^*\right| \left( \Delta^{-1}\Delta_n + \sqrt{\frac{|V_n|}{|V|}} \right) \right\|_F^2
\label{eq:opt_kmeans_graphLevel_v4} \\ 
&:= \widetilde{a}_n(\bm{X}_n,\bm{W}) + \widetilde{b}_n(\bm{X}_n) 
:=  \eta_n(\bm{X}_n,\bm{W}) \nonumber
\end{align}

\vspace{-5mm}
\noindent where 

\vspace{-9mm}
\begin{align}
  \gamma_n &= \left\|\bm{Z}\Delta^{-1}\Delta_n\bo{Z}_n^{T}\right\|_F^2 = \text{tr}\left(\bm{Z}_n\Delta_n\Delta^{-2}\Delta_n\bm{Z}_n^{T}\right) \nonumber \\
           & = \text{tr}\left(\Delta_n^2\Delta^{-2}\right)=\sum_{m=1}^K \frac{|G_{n m}|}{|G_{\cdot m}|}. 
\end{align}
\vspace{-4mm}

Recall that $G_{nk}$ is the set of nodes from $n$ that are in cluster $k$, and $G_{.k}$ is the set of nodes from all graphs in cluster $k$. 
The last inequality \eqref{eq:opt_kmeans_graphLevel_v3} uses the fact that $\bm{Z}$ and $\bm{Z}_n$ are orthogonal matrices, %
we demonstrate why the following inequality holds:
\begin{align*}
\left\| \bm{X}_{n}\bo{W} \!\!-\! \bo{Q}_n^* \right\|_F^2 \gamma_n + \left\|\bo{Q}_n^*\bm{Z}^T\Delta^{-1}\Delta_n\bm{Z}_n - \bo{Q}_n^*\sqrt{\frac{|V_n|}{|V|}} \right\|_F^2 \le \left\| \bm{X}_{n}\bo{W} \!\!-\! \bo{Q}_n^* \right\|_F^2 \gamma_n + \left\| \left|\bo{Q}_n^*\right|\Delta^{-1}\Delta_n + \left|\bo{Q}_n^*\right| \sqrt{\frac{|V_n|}{|V|}} \right\|_F^2   
\end{align*}
First, we drop what is constant on both sides of the inequality first term on both sides of the inequality.
\begin{align}
&\left\|\bo{Q}_n^*\bm{Z}^T\Delta^{-1}\Delta_n\bm{Z}_n - \bo{Q}_n^*\sqrt{\frac{|V_n|}{|V|}} \right\|_F^2  \le \left\| \left|\bo{Q}_n^*\right|\Delta^{-1}\Delta_n + \left|\bo{Q}_n^*\right| \sqrt{\frac{|V_n|}{|V|}} \right\|_F^2   
\label{eq:ineq_imp}
\end{align}
Rewriting LHS of Eq. \eqref{eq:ineq_imp} using trace operator, we have
\begin{align}
&\left\|\bo{Q}_n^*\bm{Z}^T\Delta^{-1}\Delta_n\bm{Z}_n - \bo{Q}_n^*\sqrt{\frac{|V_n|}{|V|}} \right\|_F^2 = tr\left(\Delta_n^2\Delta^{-2}\bm{Z}\bo{Q}_n^{*T}\bo{Q}_n^*\bm{Z}^T\right)  \nonumber\\
&\qquad \qquad- \sqrt{\frac{|V_n|}{|V|}}tr\left(\bm{Z}_n^T\Delta_n\Delta^{-1}\bm{Z}\bo{Q}_n^{*T}\bo{Q}_n^*\right)  - \sqrt{\frac{|V_n|}{|V|}}
tr\left(\bo{Q}_n^{*T}\bo{Q}_n^*\bm{Z}^T\Delta^{-1}\Delta_n\bm{Z}_n\right) +\frac{|V_n|}{|V|}tr\left(\bo{Q}_n^{*T}\bo{Q}_n^*\right)
\label{eq:ineq_lhs}
\end{align}
The RHS is given by 
\begin{align}
&\left\| \left|\bo{Q}_n^*\right|\Delta^{-1}\Delta_n + \left|\bo{Q}_n^*\right| \sqrt{\frac{|V_n|}{|V|}} \right\|_F^2 = \nonumber\\
&=tr\left(\Delta_n^2\Delta^{-2}|\bo{Q}_n^{*T}||\bo{Q}_n^*|\right) + 2\sqrt{\frac{|V_n|}{|V|}} tr\left(\Delta_n\Delta^{-1}|\bo{Q}_n^{*T}||\bo{Q}_n^*|\right)  +\frac{|V_n|}{|V|}tr\left(|\bo{Q}_n^{*T}||\bo{Q}_n^*|\right)
\label{eq:ineq_rhs}
\end{align}
Finally, the inequality in Eq.~\eqref{eq:ineq_imp} holds because the difference between the RHS and the LHS is positive. Given $\bm{Z}$ is orthonormal, we have 
\begin{align*}
2\sqrt{\frac{|V_n|}{|V|}} tr\left(\Delta_n\Delta^{-1}|\bo{Q}_n^{*T}||\bo{Q}_n^*|\right)+ tr\left(\Delta_n^2\Delta^{-2}|\bo{Q}_n^{*T}||\bo{Q}_n^*|\right)  =&2\sqrt{\frac{|V_n|}{|V|}} tr\left(\Delta_n\Delta^{-1}|\bo{Q}_n^{*T}||\bo{Q}_n^*|\right)+ \left\||\bo{Q}_n^*|\Delta_n^1\Delta^{-1}|\right\|_F^2 \\
&\ge tr\left(\Delta_n^2\Delta^{-2}\bm{Z}\bo{Q}_n^{*T}\bo{Q}_n^*\bm{Z}^T\right)
\end{align*}
\hfill  \qquad \qquad \qquad \qquad \qquad \qed 

\subsection{Proof of Lemma \ref{lemma:same_prop_pop}}
\label{subsec:proof_lemma}

\textit{\textbf{Proof}}
For graph $n$, the vector of counts of nodes in each cluster has expected value given by $\mathbb{E} \left[\sum_i^{|V_n|}\bm{X}_{ni}\right] = |V_n|\bm{\zeta}$. 
Assuming the same distribution of the nodes over cluster for all graphs, $\mathbb{E}\left[\sum_n^{N}\sum_i^{|V_n|}\bm{X}_{ni}\right] = |V|\bm{\zeta}$. 
We know that $\sum_i^{|V_n|}\bm{X}_{ni} = \diag(\bm{X}_n^{T}\bm{X}_n) = \diag(\Delta_n^2)$ and $\sum_i^{|V|}\bm{X}_{ni} = \diag(\bm{X}^{T}\bm{X}) = \diag(\Delta^2)$. 
Defining $\alpha_n = |V_n|/|V|$ for all $n \in [1,...,N]$, we have 
$$
\mathbb{E}\left[\Delta_n\Theta\Delta_n\right]=\sqrt{\alpha_n}\mathbb{E}\left[\Delta\Theta\Delta\right]\sqrt{\alpha_n}
 $$
Note that if all graphs have the same size, $|V_n|$, then $\alpha_n=N^{-1}$. Furthermore, using the eigendecomposition on both sides, we have

\vspace{-6mm}
$$
\mathbb{E}\left[\bm{Z}_n \bm{D}_n \bm{Z}_n^{T}\right]=\sqrt{\alpha_n}\mathbb{E}\left[\bm{Z} \bm{D} \bm{Z}^{T}\right]\sqrt{\alpha_n}
$$
Thus,
$
\bm{Z}_n = \bm{Z} \iff \bm{D}_n = \alpha_n\bm{D}
$.

Finally, 
\vspace{-3mm}
\begin{align*}
\mathbb{E}\left[\bm{Z}_n^{T}\Delta_n^{-1}\Delta\bo{Z}\right] &=  \mathbb{E}\left[\bm{Z}^{T}\left(\alpha_n^{-1/2}\Delta^{-1}\right)\Delta\bo{Z}\right] \nonumber \\
&= \mathbb{E}\left[\alpha_n^{-1/2}\bm{Z}^{T}\bm{Z}\right]=\alpha_n^{-1/2}%
\end{align*}
\hfill  \qquad \qquad \qquad \qquad \qquad \qed 

\subsection{Proof of Lemma \ref{lemma:same_obj}}
\label{subsec:proof_lemma_2}

\textit{\textbf{Proof}} By Lemma \ref{lemma:same_prop_pop},\\
\vspace{-4mm}
$$\mathbb{E}[\eta_n(\bm{X}_n,\bm{W})] \propto\left\| \bm{X}_{n}\bo{W} - \bo{Q}_n^* \right\|_F^2\frac{|V_n|K}{|V|}$$ 
where  $\bo{Q}_n^*=\bo{Q}_n\frac{|V|}{|V_n|}=\bm{U}_n\bm{D}\frac{|V_n|}{|V|}\frac{|V|}{|V_n|}=\bm{U}_n\bm{D}$. Given $\frac{|V_n|}{|V|}=\frac{1}{N}$ and dropping all constants across graphs, we have $\mathbb{E}[\eta_n(\bm{X}_n,\bm{W})] \propto \|\bm{X}_{n} \bm{W}_n - \widehat{\bm{U}}_n\|_F^2$
\hfill  \qquad \qquad \qquad \qquad \qquad \qed

\subsection{Toy data example}
\label{subsec:toydata}

As an example to illustrate the effect of $\bm{Z}_n^{T}\Delta_n^{-1}\Delta\bm{Z}$ on the individual eigendecomposition, consider three graphs, where each graph is a village, nodes represent individuals and edges represent relationships between them. Assume that individuals are clustered in four different blocks based on their personalities, which reflects how they form relationships. 
Figure \ref{fig:toydataconnectivity}(left) shows the connectivity matrix based on those personalities. Also, consider that each village has its own distribution of people over the clusters, as shown, for instance, in Figure  \ref{fig:toydataconnectivity} (right). 

\begin{figure}[htbp]
\begin{minipage}{.45\textwidth}
\begin{figure}[H]
  \centering
  \includegraphics[width=\textwidth]{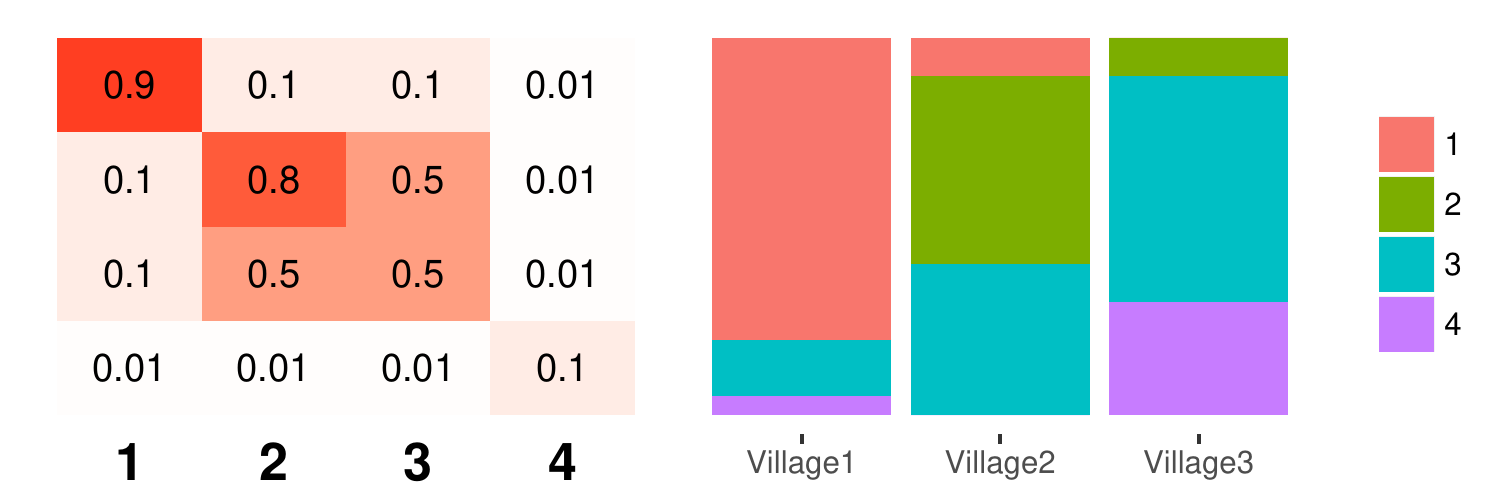}
\vspace{-2mm}
  \caption{Toy data connectivity matrix (left) and distribution of blocks per village (right). 
  }
  \label{fig:toydataconnectivity}
\end{figure}
\end{minipage}
\hspace{5mm}
\begin{minipage}{.45\textwidth}
\begin{figure}[H]
  \centering
  \includegraphics[width=\textwidth]{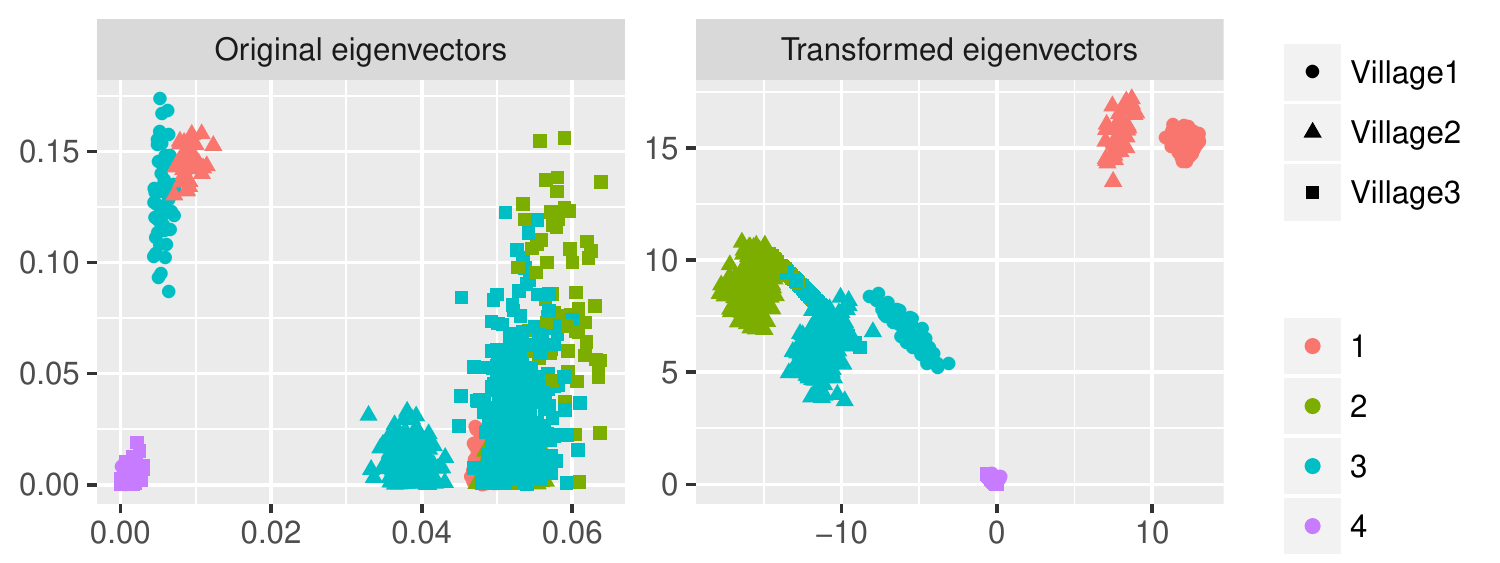}
\vspace{-2mm}
  \caption{Original and transformed eigenvectors}
  \label{fig:spec_all}
\vspace{-4mm}
\end{figure}
\end{minipage}
\end{figure}

Now, using the adjacency matrix eigendecomposition for each village $\bo{A}_n =  \widehat{\bo{U}}_n\widehat{\bo{D}}_n\widehat{\bo{U}}_n^{T}$, we have  $\widehat{\bo{Q}}_n = \widehat{\bo{U}}_n\widehat{\bo{D}}_n$ as a {proxy}  for $\bo{Q}_n$. Since  $\widehat{\bo{U}}_n\widehat{\bo{D}}_n\widehat{\bo{U}}_n^{T}=\widehat{\bo{T}}_n\widehat{\bo{D}}_n\widehat{\bo{T}}_n^{T}$ for $\bm{T}_n = -\bm{U}_n$, we consider $|\widehat{\bo{Q}}_n|$ instead. Figure \ref{fig:spec_all} (left) shows the two largest absolute eigenvectors of each adjacency matrix. Each block has a different center of mass depending on which village (graph) it is in. This is due the fact that villages have completely different distribution of nodes over the blocks. Therefore, sharing information across villages is fundamental not only to assign 
underrepresented nodes to the correct block, but also to map the blocks across villages. Using our proposed transformation given in Equation \eqref{eq:global_local_expression}, we obtain the results shown on Figure \ref{fig:spec_all} (right). We re-scale and rotate the eigenvectors in order to have an embedding of the nodes that is closer to the global eigendecomposition. Therefore, using any clustering algorithm one can correctly recover the membership for nodes across the three villages. 
\subsection{Consistency}
\label{sec:consist}
The global parameter $\bm{\Theta}$ is central in order to understand Multi-graph settings. It not only gives an overall summary of how nodes connect based on their communities, but $\bm{\Theta}$ also has the role to link all graphs. In other words, a reliable estimate of $\bm{\Theta}$ means we can predict edges between nodes in different graphs with confidence. Here, we discuss the asymptotic behavior of $\bm{\Theta}$ in Multi-graph joint SBM as $N\rightarrow \infty$. And, we show that the estimator $\widehat{\bm{\Theta}}=N^{-1}\sum_n\widehat{\bm{\Theta}}_n$ in Eq.~\eqref{eq:theta_estimator_n} converge to the global $\bm{\Theta}$ almost surely for when we know the true membership. Lemma \ref{lemma:consistency_theta} (below) formalizes these statements. 

\begin{lemma}
Let the pair $(\bm{X},\bm{\Theta})$ parametrize a SBM with $K$ communities for $N$ graphs where $\bm{X}$ contains the membership matrix of all graphs stacked and $\bm{\Theta}$ is full rank. Write $\nu \le \min_n |V_n|$. Now if assume $(\widehat{\bm{W}},\widehat{\bm{X}})$ is the optimal solution of Eq.~\ref{eq:opt_kmeans_graphLevel} then $\widehat{\bm{\Theta}}$ converge to $\bm{\Theta}$ in probability Eq.~\eqref{eq:theta_weak_consistent}. If we also assume $\widehat{\bm{X}}_n = \bm{X}_n$ then $\widehat{\bm{\Theta}}$ converge to $\bm{\Theta}$ almost surely Eq.~\eqref{eq:theta_consistent}.
\vspace{-8mm}
\begin{multicols}{2}
\begin{equation}
\lim_{\nu\rightarrow \infty;N\rightarrow \infty} \widehat{\bm{\Theta}} \stackrel{P}{\rightarrow }\bm{\Theta}
\label{eq:theta_weak_consistent} 
 \end{equation}
\begin{equation}
\lim_{N\rightarrow \infty} \widehat{\bm{\Theta}} \stackrel{a.s.}{\rightarrow }\bm{\Theta}
\label{eq:theta_consistent} 
 \end{equation}
\end{multicols}
\label{lemma:consistency_theta}
\end{lemma}
%
%
%
%
%
%
%
%
%

%
\vspace{-4mm}
\textit{\textbf{Proof}} Eq.\eqref{eq:theta_weak_consistent} follows directly from the fact that $\bm{A}_n$ converge in probability to $\bm{P}_n$ for large $|V_n|$, Theorem 5.2~\cite{lei2015consistency}. Eq.\eqref{eq:theta_consistent} follows from Eq.~\eqref{eq:expected_estimator_n}, we know that $\mathbb{E}[\widehat{\bm{\Theta}}_n] = \bm{\Theta}$. 
Thus, using Kolmogorov-Khintchine strong law of large numbers,  $\widehat{\bm{\Theta}} \rightarrow \bm{\Theta}$ almost surely for large $N$. 

\vspace{-2mm}
$\hfill \qquad \qquad \qquad \qquad \qed$

Eq.~\eqref{eq:theta_consistent} is only true in the Multi-graph joint case. In the Isolated setting, we need to \textit{re-align} the memberships across graphs which adds an extra layer of complexity. For instance, assume the re-alignment procedure consists on (1) rank each community on each graph based on  $diag(\bm{\Theta}_n)$, then (2) re-order the connectivity matrix and membership accordingly. In this case, $Var(\widehat{\bm{\Theta}}) = \sum_n^N Var(\bm{\Theta}_n) \rightarrow \infty$, unless we assume graph size to be large, i.e. $\nu\rightarrow \infty$ where $ \nu \le \min_n|V_n|$. Nevertheless, this gives weak consistency at most. In fact, this is true for any realignment procedure whose performance is a function of graph size. Figure \ref{fig:syn_theta_dist} in the Synthetic experiments shows that the Joint model estimates $\bm{\Theta}$ well even for small graph settings which is not true for Isolated models.

\subsection{Additional re-alignment procedures and cluster retrieval performance measures}
\label{subsec:realign_measures}
Here, we assess the performance of \isospec   ~using more complex re-alignment procedures (Iso2 and Iso3), we compare them with the one we used in the main body of the paper (Iso1). We also include adjusted rand index (ARI) \citep{rand1971objective}  and misclutering rate (MCR) measures which are often used to evaluate cluster performance when ground truth is known.  Now, we describe the re-alignment procedures:  
\begin{itemize}
    \item Iso1 (used in the main body of the paper)
    \begin{enumerate}
        \item  Rank diag$(\theta_n)$
        \item Re-order $X_n$ and $\theta_n$ accordingly.
    \end{enumerate}
    \item Iso2
    \begin{enumerate}
        \item Cluster the centers $W_n$ across graphs
        \item Re-assign nodes to clusters based on the center of the centers.
    \end{enumerate}
 \item Iso3
 \begin{enumerate}
     \item Search over all permutation to make all $W_n$ closest possible from each other
     \item Re-order $X_n$ and $\theta_n$ accordingly.
 \end{enumerate}
\end{itemize}

We compute ARI, MCR and NMI for different re-alignment procedure for increasing heterogeneity the same synthetic data presented in \ref{fig:syn_results1}(right). In order to make these additional experiments comparable with \ref{fig:syn_results1}(right), we also include results for \model.

\begin{figure}[h]
    \centering
    \includegraphics[width=.95\textwidth]{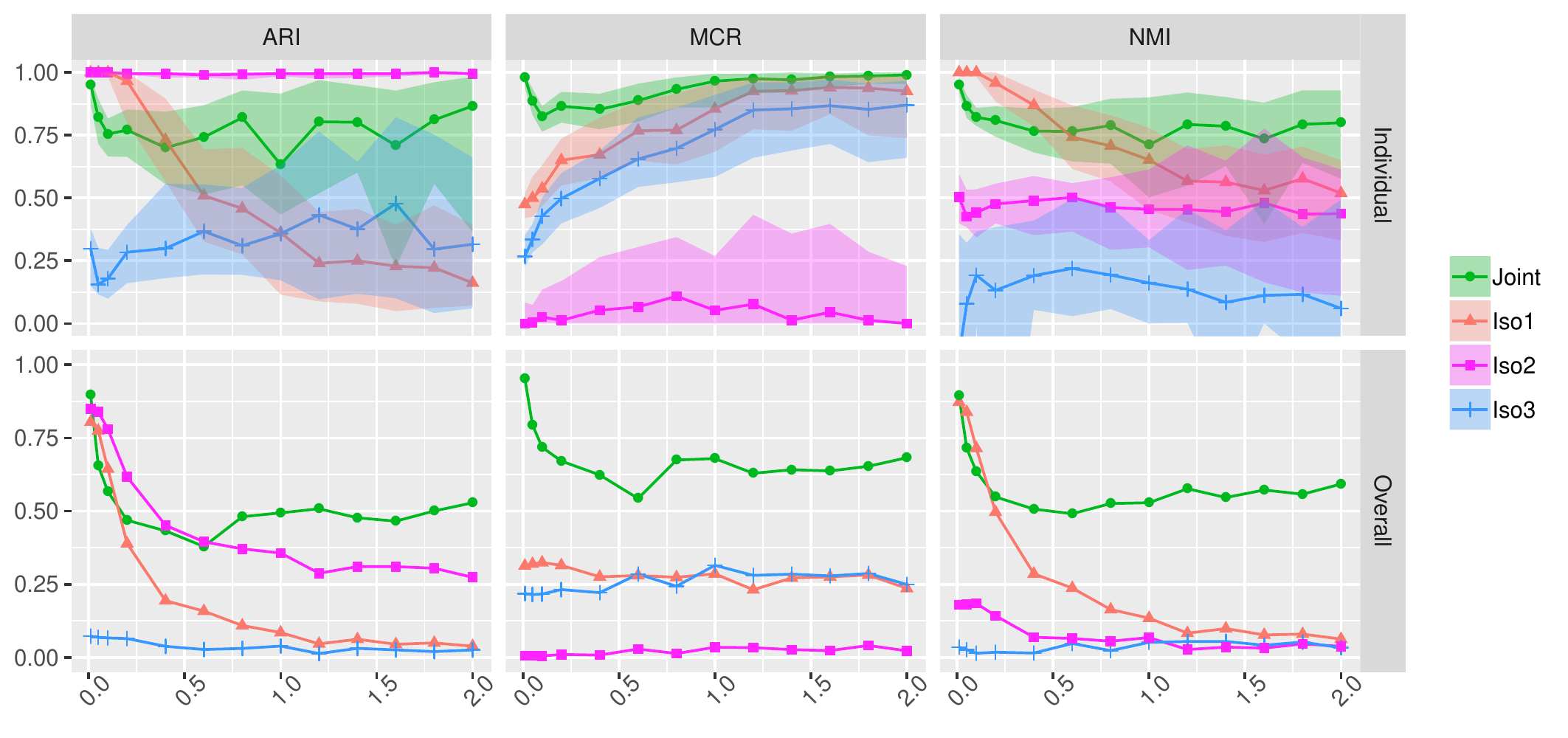}
    \caption{Fixed $N= 100$ and $|V_n|= 200$: Cluster retrieval performance curves for each measure (ARI,MCR and NMI) for each model for increasing heterogeneity ($\alpha$). Top row: median and the interquartile range curves of the individual graphs performance. Bottom row: overall curves.}
    \label{fig:realign_measure}
\end{figure}

Figure \ref{fig:realign_measure} shows that ARI and MCR have similar results of NMI for the Joint and Iso1. In terms of the additional re-alignment procedures, Iso2 seems to outperform all other methods for the Individual case if the focus is on ARI. However, it does not have a good performance using other measures, and the reason is due to the fact that Iso2 tends to accumulate nodes in lower number of clusters, i.e. $<<K$. In terms of the Overall curves, the \model outperforms all the baselines for heterogeneous settings as it was expected. 

\subsection{Synthetic data experiment for varying graphs sizes}
\label{subsec:experiments_3}
Here, we aim to assess cluster retrieval performance in settings where the graphs have different sizes. We used $|V_n|\stackrel{iid}{\sim}NB(\mu,r)$ to sample the size of each graph, where $\mu$ is the mean and $r$ the dispersion parameter. We fixed $\mu=200$ and we vary $r \in [1,10]$. Lower values of $r$ mean more variability in graph size distribution. We also consider two main scenario: 1) homogeneous $\alpha = \frac{1}{K}$; and 2) heterogeneous $\alpha = 1$. Figure~\ref{fig:experiment3} shows the curves for each model in each scenario. For the heterogeneous scenario it is clear that the JointSpec outperform the baselines. Also, NMI curves look flat for increasing $r$ on both scenarios (homogeneous and heterogeneous) which suggests that the distribution of nodes over clusters (controlled by $\alpha$) is more detrimental for cluster retrieval than the size of the graphs. 
\begin{figure}[h]
\begin{minipage}{.45\textwidth}
\begin{figure}[H]
  \centering
  \includegraphics[width=\textwidth]{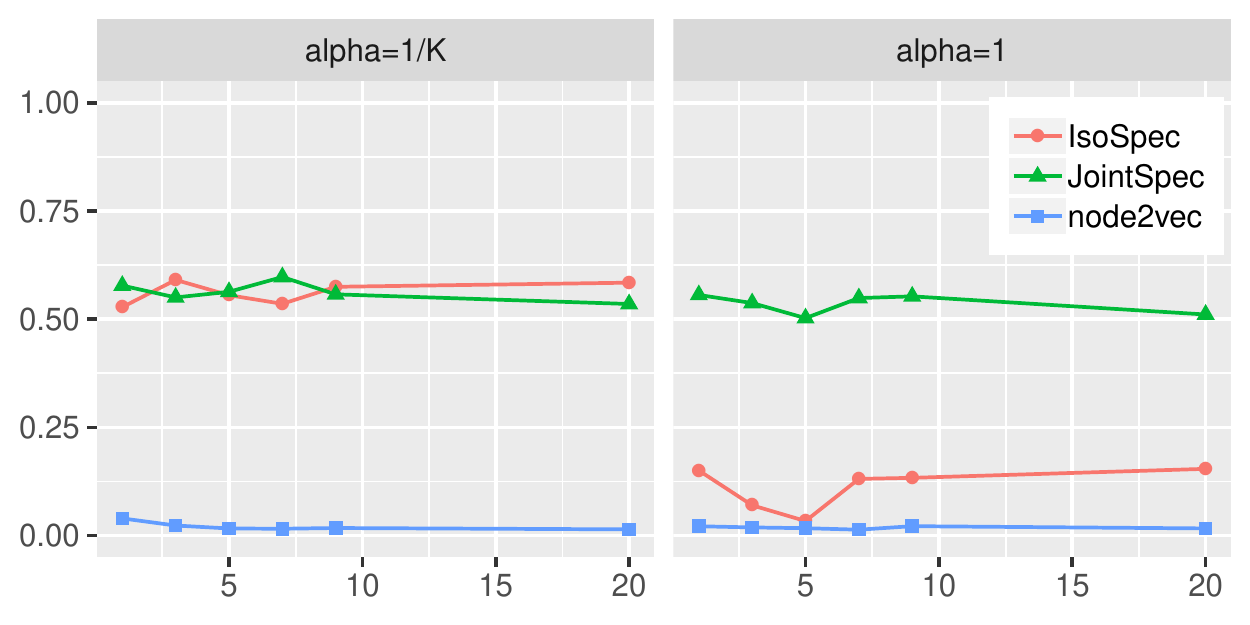}
\vspace{-3mm}
  \caption{NMI curves over $r$ where $r$ is the dispersion parameter in $|V_n|\stackrel{iid}{\sim}NB(\mu,r)$ for homogeneous ($\alpha=1/K$) and heterogeneous ($\alpha=1$) scenarios}
  \label{fig:experiment3}
\end{figure}
\end{minipage}
\hspace{5mm}
\begin{minipage}{.45\textwidth}
\begin{figure}[H]
	\centering
	\includegraphics[width=\textwidth]{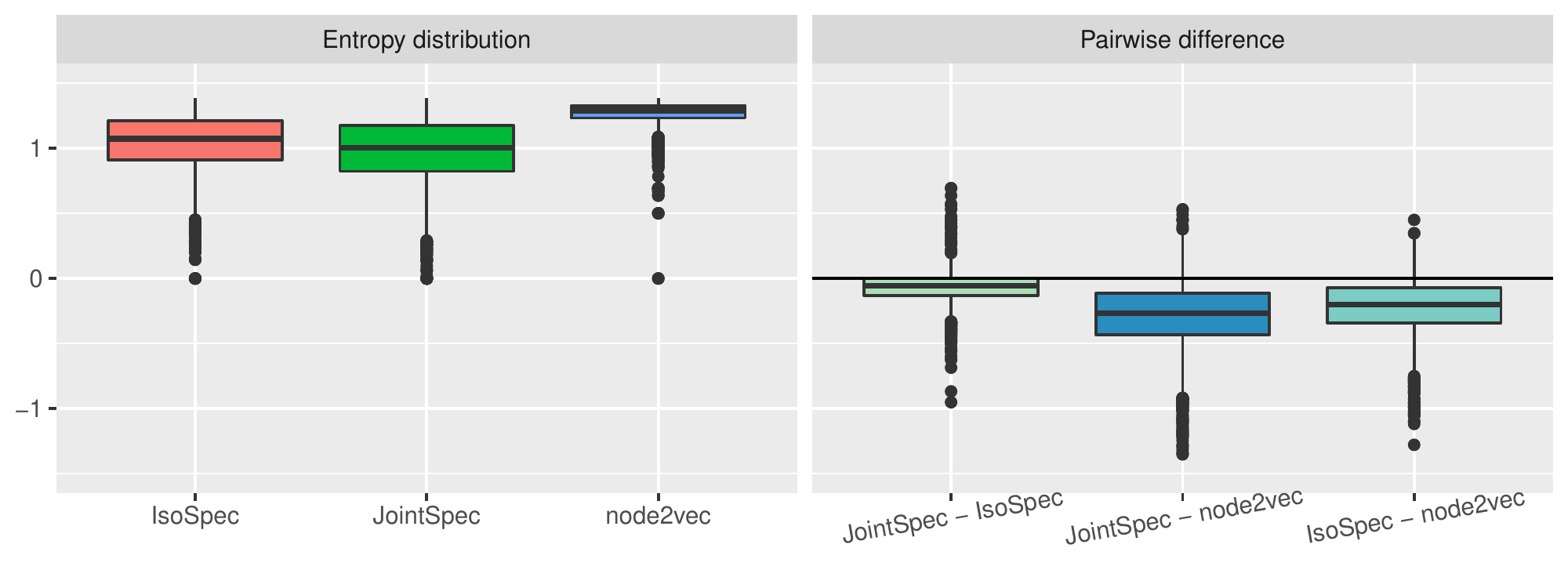}
\vspace{-4mm}
	\caption{Entropy distribution of the words per model (left) and pairwise entropy difference between models for each word (right).}
	\label{fig:twitter_entropy}
\vspace{-4mm}
\end{figure}
\end{minipage}
\end{figure}

\subsection{Assessment of community assignments in Twitter}
\label{subsec:experiments_4}

We also compute the entropy of the community assignment per words across graphs. We expect the community of the words to be consistent across graphs, therefore a lower entropy.
We found that the Joint model had the lowest entropy overall, Figure \ref{fig:twitter_entropy} shows the results. Also, we include a pairwise distribution of the difference of the entropies which shows that JointSpec had the lowest entropy for the majority of the words.

\end{document}